\documentclass[journal]{IEEEtran}

\usepackage{graphicx}
\usepackage{amssymb}
\usepackage{amsmath}
\usepackage{color}
\usepackage{subfigure}
\usepackage{epstopdf}
\usepackage{graphicx}
\usepackage{chngcntr}
\usepackage{epstopdf}

\newcommand{\qed}{\hbox{\rule[-2pt]{3pt}{6pt}}}

\newcommand{\be}{\begin{equation}}

\newcommand{\eq}{\end{equation}}
\newcommand{\ba}{\begin{array}}
\newcommand{\ea}{\end{array}}
\newcommand{\bean}{\begin{eqnarray*}}
\newcommand{\eean}{\end{eqnarray*}}
\newcommand{\bea}{\begin{eqnarray}}

\newcommand{\eea}{\end{eqnarray}}
\newcommand{\nn}{\nonumber}

\newcommand{\R}{\rm {I\kern-2pt R}}

\newcommand{\beq}{\begin{equation}}
\newcommand{\eeq}{\end{equation}}

\newcommand{\blm}{\begin{lemma}}\newcommand{\elm}{\end{lemma}}
\newcommand{\bit}{\begin{itemize}}\newcommand{\eit}{\end{itemize}}

\newtheorem{theorem}{Theorem}[section]
\newtheorem{remark}[theorem]{Remark}
\newtheorem{lemma}[theorem]{Lemma}
\newtheorem{example}[theorem]{Example}
\newtheorem{assumption}[theorem]{Assumption}
\newtheorem{corollary}[theorem]{Corollary}

\newtheorem{definition}[theorem]{Definition}

\ifCLASSINFOpdf
\else

\fi

\begin{document}

\title{
\huge{Compensation of Input/Output Delays
for  Retarded Systems by Sequential Predictors:  
A Lyapunov-Halanay Method
}}

\author{Xin Yu  \    and \   Wei Lin, {\it Fellow, IEEE}
\thanks{The authors are
with the Dept. of Electrical Eng. and Computer Science, Case Western Reserve University,
Cleveland, Ohio, USA.
}
}


\maketitle

\begin{abstract}
This paper presents a Lyapunov-Halanay method to study   
global asymptotic stabilization (GAS) of  
nonlinear retarded systems subject to large constant delays in input/output  --- a challenging problem due to their inherent destabilizing effects.
Under the conditions of global Lipschitz continuity (GLC) 
and global exponential stabilizability (GES) of the retarded  
system without input delay, a state feedback controller is designed based on sequential predictors
 to make  the closed-loop retarded system GAS. Moreover, if the 
 retarded system with no output delay permits a global exponential observer,
    a dynamic  output compensator is also constructed based on sequential predictors, achieving GAS of the
corresponding closed-loop retarded system with input/output delays.
The predictor based state and output feedback stabilization results 
are then extended   to a broader class of nonlinear retarded systems with input/output delays, which 
 may not be  GES  
 but satisfy  global asymptotic stabilizability/observability  and suitable ISS conditions. 
 As an application,   a pendulum system with delays in the state, input and output  is used to illustrate 
 the effectiveness of the proposed state and output feedback control strategies
  based on sequential predictors. 
\end{abstract}


\IEEEpeerreviewmaketitle

\section{Introduction}
\counterwithout{equation}{section}

The control of nonlinear dynamical systems with delays in the state, input, and output has 
been one of  central research themes 
 in the field of  time-delay systems. 
On one hand, time-delays are  an inevitable phenomenon inherent in the actuators, sensors, and data transmission processes 
of  physical 
systems.  On the other hand,  the delays often exert adverse effects on the system's 
performance and stability \cite{kr63,gu03,nm14}.


For time-invariant linear systems that possess input 
delay, the 
seminal paper 
\cite{art82} introduces a reduction method to compensate the large input delay effectively 
by using a predictor-based controller, which is indeed a  
generalization of the well-known Smith predictor \cite{sp59}. 
In the time-varying case, 
 \cite{mmn14}  further  
 applied  the reduction model approach to address the asymptotic stabilization problem for
 time-varying linear systems with input delay. 
Over the years,  various  predictor based feedback control strategies, both exact and approximate, 
have been developed   
(e.g.,  for linear strict-feedback systems with delayed integrators \cite{lk10}),  
 and more significantly extended to input-delayed nonlinear systems, 
 thanks to the original work
\cite{krs10}, where a prediction-based control scheme is proposed to compensate
 for large input delays in nonlinear systems that are forward complete 
 as well as  strict-feedforward nonlinear systems that are free of state delays.  
 
 Since then, many subsequent works have been reported in the literature and substantial progress has been obtained.
 For instance, the paper \cite{kra11} developed a stabilization method via approximate predictor-based feedback
  for nonlinear systems subject to input delay.  Under the globally Lipschitz continuous  (GLC) condition, 
  \cite{kk13} 
  further  developed an output feedback control strategy  to compensate long input/output delays 
 in disturbed nonlinear systems, by virtue of  
  approximate predictors and high-gain observers. 
For time-delay nonlinear systems, 
 the work  
\cite{BL14}  has filled a critical gap by solving the delay compensation problem
of nonlinear systems with arbitrary-length input and state delays.   
 For additional  details,
  the reader is referred to the  monograph \cite{kk17} that has documented comprehensively the state-of-the-art  
     in  the development 
   of   predictor-based designs and compensation 
   methods  for  
   nonlinear systems with  large  input/output delays.

The aforementioned  predictor-based  
compensation 
approaches  typically involve an infinite-dimensional term, namely, 
an integral term in the nonlinear predictor, which may give rise to practical implementation challenges. 
       To address the practical feasibility and implementation issue,  alternative methodologies emphasizing simplicity and robustness have been explored. In particular,  attractive delay-free 
        feedback control schemes  have been sought  by a number of researchers.
        For instance,   the early papers \cite{mmn03,Mazenc04} presented  memoryless nested saturation feedback controllers to 
         achieve global asymptotic stabilization for a chain of linear integrators and  feedforward systems with  input delay,
         respectively, while the recent works \cite{sunlin21,sunlin23} developed delay-free,  
         {\it dynamic gain-based} saturation control schemes 
         for  a chain of delayed integrators perturbed by  feedforward nonlinearities with large state/input delays 
          and parametric uncertainty.  Notably, the {\it dynamic state feedback} method was originally developed for the
          control of nonlinear {\it affine systems} with delays in the state 
          by delay-free (finite-dimensional)  \cite{LinZhang20} 
           instead of   delay-dependent (infinite-dimensional) feedback controllers \cite{jan}.
           
Aiming at designing delay-free feedback controllers whenever possible,
we characterized in \cite{wl22,yulin23,zhaolin24}  mild regularity conditions  for  
input-delayed general  nonlinear systems {\it without affine structure} to be globally or semigobally  
asymptotically stabilizable by memoryless state and/or output feedback. Due to the use of memoryless feedback, 
the obtained results in \cite{wl22,yulin23,zhaolin24} 
are only applicable to limited input delay. This is expected even in the case of  linear systems with input delay,  
 as illustrated by the simple example in \cite{zhaolin20}. 
         Under the linear growth  condition  that ensures the forward completeness
          of  nonlinear systems under consideration,  
  \cite{zhaolin21}   designed memoryless  feedback controllers 
   for a class of feedforward systems with large delays in the state and input,  
   More 
  recently,   \cite{YL24} developed a delay-free, non-identifier based adaptive 
  control strategy to globally stabilizes  a chain of integrators  perturbed by linearly growth feedforward systems 
    with  unknown  state/input delays and unknown parameters, 
    eliminating the need for estimating both unknown delays and parameters.
  Building on this progress,  the  work \cite{YL24b}  further developed a universal output feedback control method
   that is capable of coping with  unknown large delays in the state, input and output simultaneously, 
    for a class of time-delay feedforward systems.  
    It should be  
    stressed that the design of these memoryless controllers imposes,
     however,  stringent structural requirements on the time-delay nonlinear systems.

To circumvent the structural restrictions on  time-delay systems and the 
integral terms in the predictor based controllers,  numerous researchers have 
developed some delay compensation methods by virtue of  
sequential predictors or chain predictor.  
The underlying philosophy of these  
approaches is to construct a finite number of cascades of time-delay differential equations which estimate the state
  forward in time,  
   providing a prediction without involving 
    integral terms. In this way, they need no structural conditions on the time-delay nonlinear systems 
    but do impose the global Lipschitz continuity or the ISS conditions. 
  For example, 
    the early work \cite{bgb07} proposed a state feedback controller based on an asymptotic state prediction 
for  linear 
systems and 
 globally Lipschitz continuous (GLC) systems with input delay.
 \cite{mm17} introduced a new sequential predictors based state feedback to stabilize time-varying nonlinear  systems with input delay, which are GLC with respect to state variables and have an appropriate ISS property.
 The paper  \cite{caca16} studied feedback control using a cascade of predictors 
for a class of strict-feedback systems with large delays in the input and output  under certain GLC assumptions. 
In \cite{caca17},  a chain of predictors is proposed for stabilization of linear systems 
subject to  input, state, and output delays. In the case of time-varying linear systems,
 sequential predictors have also been designed under time-varying input delay and sampling \cite{wm19},  
 and continuous-discrete sequential observers were presented in \cite{mm20}. For bilinear systems with input delay, 
 \cite{bm21} presented a delay compensation scheme by designing a
 feedback stabilizer based on sequential predictors under a suitable ISS condition.
More recently,   \cite{pepe25} has proposed a chain of observer-based, GLC output feedback controller to compensate 
 constant input/output delays of a class of  retarded nonlinear systems that are GLC. Interestingly,  
  the analysis and synthesis approach given
 in \cite{pepe25} relies on neither the Lyapunov-Krasovskii functional method nor the Razumikhin theorem.
  Instead, it is based on a solution estimation of the retarded closed-loop system and an open-loop chain  
  observer based controller.   However, the method  works only for the zero initial condition of the observers and thus 
  raises a question on its effectiveness when facing any non-zero initial condition  as well as the
  robustness  issue  with respect to the initial condition of the proposed chain observers.

In this article, we develop a Lyapunov-Halanay method to address 
the GAS problem  
for nonlinear retarded systems subject to constant delays in input/output, by either state or output feedback 
based on sequential predictors.
Under the GLC condition  
and global exponential stabilizability of the retarded  
control system without input delay, 
we give a solution to the GAS problem via state feedback by designing 
a state feedback controller based on sequential predictors,  rendering
 the retarded closed-loop  system GAS.   In the case of output feedback, 
 we show that the GAS problem is also solvable by a dynamic output feedback controller 
 based on sequential predictors,   provided that the nonlinear 
 retarded system without output delay permits an global exponential observer.
Finally, we extend the obtained GAS results again by means of  
sequential-predictor based state and output feedback, respectively,
to a wider class of nonlinear retarded systems with input/output delays, which   
only  satisfy  global asymptotic stabilizability/observability and suitable ISS conditions.

The  paper is organized as follows: 
In Section II,   standing assumptions and some key lemmas are presented. 
Section III elaborates the design of state feedback controller based on  sequential predictors, including the stability analysis that establishes the global asymptotic stability of the closed-loop retarded system. Section IV gives a counterpart of state feedback,
i.e.,  an output feedback controller based on sequential predictors, 
as well as the stability analysis of the closed-loop retarded systems. 
Section V presents  extensions to  ISS retarded systems with input/output delays.  As an application,  
a pendulum system with state
and input/output delays is considered in Section VI, and simulation results  are given  to verify the effectiveness of the 
control laws.  Conclusions are given in Section VII.

 {\it Notation:}  Let $\mathbb{N}$ denote the set of 
natural numbers  and $\R^n$ be the $n$ dimensional Euclidean space. For $x\in \R^n$, $\|x\|$  
 is the Euclidean norm of $x$.
 Let  $C([-\delta,0], \R^n)$   
be  a space of continuous functions mapping from $[-\delta, 0]$ into $\R^n$, which is
 the Banach space 
 endowed with the supremum norm  $\|\cdot\|_c$ defined by $\|\phi\|_c=\sup_{\theta\in[-\delta,0]}\|\phi(\theta)\|$  for $\phi\in C([-\delta,0], \R^n)$. Throughout this paper, for $x_t\in C([-\delta,0], \R^n)$, $x_t(\theta)=x(t+\theta)$, $\forall\theta\in [-\delta,0]$. The
 notations
$z_{t}^i,e_{t}^i\in C([-\delta,0], \R^n)$ stand for 
$z_{t}^i(\theta):=z_{i}(t+\theta), e_{t}^i(\theta):=e_{i}(t+\theta)$,  $\forall\theta\in [-\delta,0]$,
$i\in \mathbb{N}$, and {\it are not the $i$-th power of the functionals $z_t$, $e_t$}.

\section{Standing Assumptions and Key Lemmas}

In this paper we consider a class of  nonlinear retarded 
systems with input and output delays described by the functional differential equation
\bea
&&\hskip -.2in\dot{x}(t) = f(x_t, u(t-d)) \label{s1} \\
&&\hskip -.2in 
y(t)=h(x(t-\tau)), \label{y}
\eea
where $x_t \in C([-\delta,0], \R^n)$ is the system state denoted as the history of segment at time $t$
and defined by $x_t(\theta)=x(t+\theta)$ for $\theta\in[-\delta,0]$,  $u(t) \in \R^p$  and $y(t)\in \R^q$ are the input and output at time instant $t$. 
 The positive constants $\delta ,d, \tau$ denote the time-delays in state, input and output, respectively, 
 which are known and can be arbitrarily large.  
The mappings $f: C([-\delta,0], \R^n) \times \R^p \to \R^n$ and $h: 
\R^n \to \R^q$ are globally Lipschitz continuous (GLC) with  $f(0,0) = 0$ and  $h(0) = 0$, i.e., there are 
 constants $L_f,L_h>0$, such that
\bea 
&&\hskip -.2in 
\| f(\phi, u) - f(\varphi, v)\| \leq L_f ( \| \phi- \varphi \|_c  +  \| u -v \|  ) \label{Lf} \\
&&\hskip -.2in 
\| h(x) - h(z)\| \leq L_h  \|x- z\|, \label{Lh}
\eea
 $\forall (\phi, u),  (\varphi, v) \in C([-\delta,0], \R^n) 
\times \R^p$ and  
$\forall x,z \in \R^n$. 
The GLC condition (\ref{Lf}) ensures the forward completeness of the system 
(\ref{s1}), but also allow us to compensate the large input delay 
by predictor based state and/or output feedback.
   The initial state of the retarded system (\ref{s1}) is $x_0=\psi\in C([-\delta,0], \R^n)$ 
   and the initial input is $u(\theta)=0$ for $\theta\in[-d-\tau,0]$.

Because of  the simultaneous presence of large delays in the state and input/output,
 design of  a feedback controller that renders the retarded  system (\ref{s1}) GAS 
 has been recognized as a challenging problem.  
 The purpose of this research is to find, if possible,  
both state and output feedback controllers based on sequential predictors,
 which tackle the problem for the nonlinear retarded system (\ref{s1}) under appropriate conditions.
   
   Throughout this paper, we make the following assumptions.

\begin{assumption} \label{ap1} \rm
For the retarded system (\ref{s1}) without input delay, i.e., $d=0$,
 there exists a 
GLC state feedback law $u = \alpha(x)$ with $\alpha(0)=0$,  such that the closed-loop system 
\bea
\dot{x}(t) = f(x_t, \alpha(x(t))) \label{s2}
\eea
is globally exponentially stable (GES). That is,  
$\|x(t)\|\leq a_1\|x_0\|_ce^{-a_2t}$,  $\forall x_0\in C([-\delta,0], \R^n)$, where $a_1$ and $a_2$ are positive real constants.
\end{assumption}

\begin{assumption} \label{ap2} \rm
For the input/output  
system (\ref{s1})-(\ref{y}) without output delay,
 i.e., $\tau=0$,  there  exists a GLC functional
$F: C([-\delta,0], \R^n) \times \R^p \times \R^q\to \R^n$ with a Lipschitz constant $L_{F}>0$ and $F(0,0,0)=0$, such that
\bea 
\dot{\hat{x}}(t) = F(\hat{x}_t, u(t-d), y(t)) \label{obs}
\eea
is a global  
exponential 
observer 
of the retarded 
 system (\ref{s1}), i.e., 
$\|x(t)-\hat x(t)\|\leq \lambda_1\|x_0-\hat{x}_0\|_ce^{-\lambda_2t}$, $\forall x_0,\hat{x}_0\in C([-\delta,0], \R^n)$, where 
$\lambda_1$ and $\lambda_2$ are positive 
constants. 
\end{assumption}

\begin{remark}\rm
While Assumption \ref{ap1} is a 
 stabilizability condition for the nonlinear retarded system
(\ref{s1}) with $d=0$ (in the absence of input delay) to be GES by the GLC controller $u=\alpha (x)$,  
  Assumption \ref{ap2} is  
  the existence  condition of an global exponential observer for the input/output system (\ref{s1})-(\ref{y})
  with $\tau=0$ (in the absence of output delay), which can be viewed as a nonlinear counterpart of the 
  hypothesis for
    time-delay linear systems \cite{caca17}.
    It reduces to the observability of the pair $(A,C)$  
    when the input/output delayed system (\ref{s1})-(\ref{y}) is  linear and involved  no state delay. 
    However, in the presence of state delay, characterizing the condition for the existence of an 
    observer is a nontrivial job.   This is true even for  time-delay linear systems. 
    On the other hand, 
it worth pointing out 
  that for a class of GLC retarded systems in a strict feedback form with input/output delays, i.e., 
  $\dot x_i(t)=x_{i+1}(t)+\phi_{i}(x^{1}_t,\cdots, x^{i}_t),
  \ i=1,\cdots,n,$
   with $x_{n+1}:=u(t-d)$ and $y(t)=x_1(t-\tau)$,
   there exists a GES observer, as illustrated in Remark \ref{rk2}.   
That is,  Assumption \ref{ap2} is fulfilled, 
    \end{remark}

To address the question of compensation of large delays in the input and output,
we review a number of important lemmas to be used in the sequel. 

\begin{lemma} \label{lm1}(\cite{halanay} Halanay inequality) \rm
Let $t_0$ and $\delta \ge 0$ be  real constants.  
If a $C^1$
function $w:[t_0-\delta, +\infty)\to\R_+$  satisfies
\bea
\dot w(t) \le-aw(t)+b \max_{\theta\in[-\delta,0]}w(t+\theta), \quad t\ge t_0,
\eea
where
$a>b\ge0$ are constants.
Then, there exists a constant $\lambda>0$ satisfying the equation $\lambda+ be^{\lambda \delta}=a$,
such that 
\bea w(t) \le \max_{\theta\in[t_0-\delta, t_0]} w(\theta)e^{-\lambda (t-t_0)}, \quad  t\ge t_0. \eea
\end{lemma}

Next, we recall the notion of input-to-state stability (ISS) and related properties for the nonlinear
 retarded system 
\bea\label{gv}
\dot{x}(t)=g(x_t,v(t)), \quad x_0=\psi \in C([-\delta,0], \R^n),
\eea
where $g:C([-\delta,0];\R^n) \times \R^p  \to \R^n$ is 
Lipschitz continuous on bounded sets with $g(0,0)=0$, and  the input $v(t)\in \R^p$
 is a piecewise continuous and bounded function.

\begin{definition}(\cite{teel,pepe08,Ant23a}) \rm 
The nonlinear retarded 
system (\ref{gv}) is said to be input-to-state stable (ISS),  
 if there exist a   function $\beta (\cdot,\cdot)$ of  class $\mathcal{KL}$
  and a unction $\gamma (\cdot)$ of class $\mathcal{K}$,  such that  the solution
$x(t)$ of (\ref{gv}) exists  $\forall t\geq0$ and satisfies 
\bea\label{iss}
\|x(t)\|\leq \beta(\|\psi\|_c,t)+\gamma(\sup_{0\leq s\leq t}\|v(s)\|).
\eea
for any initial state $x_0=\psi$  and any locally essentially bounded input $v$.
\end{definition}

The following lemma is from Theorem 3.2 in  \cite{pepe08}.

\begin{lemma}\label{lm2} \rm
Assume that the functional $g:C([-\delta,0], \R^n) \times \R^p \to \R^n$ is globally Lipschitz continuous
 with $g(0,0)=0$. If the zero-input system  
\bea
\dot{x}(t)=g(x_t,0)
\eea
 is  globally exponentially stable,  
the 
retarded system (\ref{gv}) is ISS.
 \end{lemma}

The notion of ISS can be extended to a more general class of retarded systems described by  
\bea\label{gvt}
\dot{x}(t)=g(x_t,v_t), \quad x_0=\psi \in C([-\delta,0], \R^n),
\eea
with $v_t\in C([-d,0], \R^p) $ and $d>0$ being a constant,
  while the extension of  ISS to 
 infinite-dimensional systems can be found, 
  for instance,  in  \cite{mw17,mw18,Ant23b,Ant24}.

The following definition is a variant form of uniformly globally ISS introduced in \cite{teel}.

\begin{definition}
\rm The nonlinear retarded 
system (\ref{gvt}) is said to be input-to-state stable (ISS)  
if there
exist a  function $\beta$ of class $\mathcal{KL}$  and a  function $\gamma$ of class $\mathcal{K}$, 
such that the solution $x(t)$ exists $\forall t\geq0$ and 
satisfies
\bea\label{iss}
\|x(t)\|\leq \beta(\|\psi\|_c,t)+\gamma(\sup_{0\leq s\leq t}\|v_s\|_c).
\eea
for any initial state $\psi$  and any input $v_t\in C([-d,0], \R^p) $.
\end{definition}

\begin{lemma}
\label{lm3} \rm
Consider the retarded system (\ref{gvt}). 
If there exist a $C^1$ Lyapunov function $V:\R^n\to \R_+$, and 
functions $\alpha_1(\cdot),\alpha_2(\cdot)$ of class of $\mathcal{K}_\infty$, such that
\bea\label{alp}
\alpha_1(\|x\|) \leq V(x) \leq \alpha_2(\|x\|), \qquad \forall x \in \R^n,
\eea
and along the solution of (\ref{gvt}),
\bea\label{dotv}
&&\hskip -.3in \dot V(x(t)) \leq -a V(x(t))  \nn\\
&&\hskip .35in + b \max_{\theta \in [-\delta, 0]} V(x(t+\theta)) + \gamma(\|v_t\|_c),
\eea
where $a>b > 0$ are real constants  and  $\gamma(\cdot)$ is a function of class $\mathcal{K}$,
then the retarded system (\ref{gvt}) is ISS.
\end{lemma}

{\it Proof:}\  When
\bea
V(x(t))\geq \max\Big\{ \frac{2b}{a+b}\max_{\theta \in [-\delta, 0]} V(x(t+\theta)), \frac{4}{a-b}\gamma(\|v_t\|_c) \Big\} \nn
\eea
it deduces from (\ref{dotv}) that 
\bea
\dot V(x(t)) \leq -\frac{a-b}{4} V(x(t))\leq -\alpha_3(\|x(t)\|), 
\eea
where $\alpha_3(\cdot)=\frac{a-b}{4} \alpha_2(\cdot)$.  

Observe  that $\frac{2b}{a+b}<1$. By Theorem 1 of \cite{teel}, the retarded system (\ref{gvt}) is ISS.  
This completes the proof. 
\hfill\qed

\begin{lemma} \label{lm4} \rm
 Consider a retarded system composed by 
 \bea
&&\hskip -.3in \dot{x}(t)=g_1(x_t,e_t),  \quad x_0=\psi  \in  C([-\delta_1,0],\R^n) \nn\\
&&\hskip -.28in \dot{e}(t)=g_2(e_t, x_t), \quad e_0=\phi \in  C([-\delta_2,0],\R^p), \label{g2}
 \eea
where $g_1:C([-\delta_2,0], \R^n) \times C([-\delta_1,0], \R^p)  \to \R^n$ and $g_2:C([-\delta_1,0], \R^p) \times C([-\delta_2,0], \R^n)  \to \R^p$
are  Lipschitz continuous on bounded sets 
with $g_1(0,0)=0$, $g_2(0,0)=0$.

 If there exist
functions $\beta_1(\cdot,\cdot)$, $\beta_2(\cdot,\cdot)$ of class $\mathcal{KL}$  
and a function $\gamma_1(\cdot)$ of class $\mathcal{K}$,  such that
\bea
&& \|x(t)\|\leq \beta_1(\|\psi\|_c,t)+\gamma_1(\sup_{0\leq s\leq t}\|e_s\|_c), \label{b1} \\
&& \|e(t)\|\leq \beta_2(\|\phi\|_c,t), \label{b2}
\eea
then there exists a $\mathcal{KL}$ function $\beta(\cdot,\cdot)$ such that
\bea
 \|(x(t), e(t))\|\leq \beta(\|(\psi,\phi)\|_c,t).
\eea
\end{lemma}

{\it Proof:}\   See Appendix.

\section{State Feedback via Sequential Predictors
}

In this section,  we study the case of state feedback under the hypothesis that 
 the 
 state of (\ref{s1}) is measurable.
The following result constitutes one of the main contributions of this paper.

\begin{theorem}
\label{th1} \rm  
Suppose that the retarded system (\ref{s1}) with  input delay 
satisfies Assumption \ref{ap1}. Then,  the 
state feedback controller 
\bea \label{u}
u(t) = \alpha (z_m(t))
\eea
with the sequential predictors
\bea \label{sq}
&&\hskip -.33in \dot{z}_1(t) = f \big(z_{t}^1, u(t-d + \frac{d}{m} ) \big) \hskip -.02in - \hskip -.02in (L_f+1)\big(z_1 (t \hskip -.02in - \hskip -.02in \frac{d}{m} ) \hskip -.02in - \hskip -.02in x(t) \big) \nn\\
&&\hskip -.33in \dot{z}_2(t) = f \big(z_{t}^2, u(t-d + \frac{2d}{m} ) \big) \hskip -.02in - \hskip -.02in
(L_f+1) \big(z_2  (t \hskip -.02in - \hskip -.02in \frac{d}{m}  ) \hskip -.02in - \hskip -.02in z_1(t) \big) \nn\\
&&\vdots \nn\\
&&\hskip -.33in \dot{z}_m(t) = f (z_{t}^m, u(t) ) \hskip -.02in - \hskip -.02in
(L_f+1) \big(z_m  (t \hskip -.02in - \hskip -.02in \frac{d}{m}  ) \hskip -.02in - \hskip -.02in z_{m-1}(t) \big) 
\eea
renders  
the 
retarded 
system (\ref{s1}) globally asymptotically stable (GAS), 
where $z_{t}^i\in C([-\delta, 0], \R^n)$
 means  $z_{t}^i(\theta):=z_{i}(t+\theta)$ 
$\forall \theta\in [-\delta, 0]$, $i=1,\cdots,m,$ and the integer $m$
satisfies
\bea \label{m}
m >(L_f+1)^2d. 
\eea
\end{theorem}

\begin{remark}\label{rk1} 
To better understand the construction of the sequential predictor (\ref{sq}), 
we examine the case  when the retarded system (\ref{s1}) has only
 point-wise time-delay, i.e.,
$\dot x(t)=f(x(t),x(t-\delta),u(t-d))$. In this case, the vector field $f$ in  (\ref{sq}) reduces to
\bea 
 f \big(z_{t}^i, u(t-d + \frac{id}{m} ) \big)= f \big(z_i(t),z_i(t-\delta), u(t-d + \frac{id}{m} ) \big) 
\eea
for $i=1,\cdots,m$, and the sequential predictor (\ref{sq}) becomes 
\bea \label{sq''}
&&\hskip -.2in \dot{z}_1(t) = f \big(z_1(t),z_1(t-\delta), u(t-d + \frac{d}{m} ) \big) \nn\\
 &&\hskip 1.in -(L_f+1)\big(z_1 (t \hskip -.02in - \hskip -.02in \frac{d}{m} ) \hskip -.02in - \hskip -.02in x(t) \big) \nn\\
&&\hskip -.2in \dot{z}_2(t) = f \big(z_2(t),z_2(t-\delta), u(t-d + \frac{2d}{m} ) \big) \nn\\
 &&\hskip 1.in -
(L_f+1) \big(z_2  (t \hskip -.02in - \hskip -.02in \frac{d}{m}  ) \hskip -.02in - \hskip -.02in z_1(t) \big) \nn\\
&&\hskip .2in \vdots \nn\\
&&\hskip -.2in \dot{z}_m(t) = f (z_m(t),z_m(t-\delta), u(t) ) \nn\\
 &&\hskip .7in -
(L_f+1) \big(z_m  (t \hskip -.02in - \hskip -.02in \frac{d}{m}  ) \hskip -.02in - \hskip -.02in z_{m-1}(t) \big).
\eea
\end{remark}

{\bf Proof of Theorem \ref{th1}:}\ \
Denote
\bea\label{z0t}
z_0(t) = x(t), \qquad z_{t}^0 = x_t
\eea 
and  let $z_i(t)$ be the estimate of  $z_{i-1}  (t + \frac{d}{m})$ for $i=1,\cdots,m$.
Define the 
estimation errors
\bea \label{er}
e_i(t) = z_i(t) - z_{i-1}  (t + \frac{d}{m}  ), \qquad i = 1, \cdots, m.
\eea
From (\ref{s1}), (\ref{u})-(\ref{sq}) and (\ref{z0t}),
  it follows that 
\bea \label{e1}
&& \hskip -.3in \dot{e}_1(t) = 
f\big(z_{t}^1, \alpha (z_m (t-d + \frac{d}{m} ) )\big) \nn\\
&& \hskip .2in - f\big(z_{t+\frac{d}{m}}^0, \alpha (z_m (t-d + 
\frac{d}{m} ) )\big) \nn\\
&& \hskip .2in - (L_f+1)e_1 (t - \frac{d}{m} ). 
\eea
In other words, $z_1(t)$ is the estimate of $z_0(t + \frac{d}{m})=x(t + \frac{d}{m})$.

Inductively,  it is deduced from (\ref{er}), (\ref{s1}), and (\ref{u})-(\ref{sq})  that  
\bea \label{ei}
&&\hskip -.3in \dot{e}_i(t) 
= f\big(z_{t}^i, \alpha (z_m (t-d + \frac{id}{m} ) )\big)
\hskip -.03in - \hskip -.03in
 f\big(z_{t+\frac{d}{m}}^{i-1}, \alpha (z_m (t-d+ \frac{id}{m} ) )\big) 
\nn\\
&&\hskip .in - (L_f+1)e_i (t - \frac{d}{m} ) 
\hskip -.02in + \hskip -.02in
 (L_f+1)e_{i-1}(t), \ \ i=2, \dots, m
\eea
where $e_0(t)=z_0(t) - x(t)= 0$.

For the error dynamics (\ref{e1}), consider the Lyapunov function 
\bea \label{v1}
V_1(e_1(t)) =  e_1^T(t) e_1(t).
\eea 
Using  the error signals 
\bea \label{e1t}
&&\hskip -.35in  e_1(t) = 
z_1(t) - z_0(t+\frac{d}{m})\nn\\
&&  
\hskip -.2in e_{t}^1:=z_{t}^1-z^0_{t+\frac{d}{m}}=z_{t}^1-x_{t+\frac{d}{m}}
\eea
with $x_{t+\frac{d}{m}}(\theta)=x(t+\frac{d}{m}+\theta)$, $\forall \theta\in [-\delta,0]$,
we  arrive at 
\bea \label{dv1}
&&\hskip -.3in \dot{V}_1(e_1(t)) 
\leq 2L_f \| e_1(t) \|  \| e_{t}^1 \|_c - 2(L_f+1)e_1^T(t)e_1 (t-\frac{d}{m}) \nn\\  
&&\hskip .28in \leq  2L_f  \| e_{t}^1 \|_c^2  - 2(L_f+1)\|e_1(t)\|^2 \nn\\
&&\hskip .43in + 2(L_f+1)\|e_1(t)\|\|e_1(t)-e_1 (t-\frac{d}{m})\| \nn\\
&&\hskip .28in \leq - 2(L_f+1)\|e_1(t)\|^2 +2L\| e_{t}^1 \|_c ^2  \nn\\
&&\hskip .43in + 2(L_f+1)\|e_1(t)\| \int_{t-\frac{d}{m}}^t \|\dot e_1 (s)\|\mathrm{d}s.
\eea 
By the GLC property of $f$ and (\ref{e1}), the following estimation is obtained via Cauchy-Schwarz inequality.
\bea \label{e1cs}
&&\hskip -.3in  \int_{t-\frac{d}{m}}^t \|\dot e_1 (s)\|\mathrm{d}s \leq  \int_{t-\frac{d}{m}}^t (L_f+1) \max_{\theta\in[-\delta-\frac{2d}{m},0]}\|e_1 (t+\theta)\|\mathrm{d}s \nn\\
&&\hskip .65in = \frac{d}{m}(L_f+1) \max_{\theta\in[-\delta-\frac{2d}{m},0]}\|e_1 (t+\theta)\|.
\eea
This, together with $\|e_1(t)\|\leq\max_{\theta\in[-\delta-\frac{2d}{m},0]}\|e_1 (t+\theta)\|$ and (\ref{dv1}), leads to
\bea \label{dv1'}
&&\hskip -.33in \dot{V}_1(e_1(t)) 
 \leq - 2(L_f+1)\|e_1(t)\|^2 +2L_f\| e_{t}^1 \|_c ^2  \nn\\
&&\hskip .29in + \frac{2 d}{m}(L_f+1)^2\big(\max_{\theta\in[-\delta-\frac{2d}{m},0]}\|e_1 (t+\theta)\| \big)^2  \nn\\
&&\hskip .13in = - 2(L_f+1)\|e_1(t)\|^2 +2L_f\| e_{t}^1 \|_c ^2  \nn\\
&&\hskip .29in + \frac{2 d}{m}(L_f+1)^2\max_{\theta\in[-\delta-\frac{2d}{m},0]}\|e_1 (t+\theta)\|^2  \nn\\
&&\hskip .0in \leq - 2(L_f+1)\|e_1(t)\|^2 +2L_f\max_{\theta\in[-\delta-\frac{2d}{m},0]}\|e_1 (t+\theta)\|^2  \nn\\
&&\hskip .29in + \frac{2 d}{m}(L_f+1)^2\max_{\theta\in[-\delta-\frac{2d}{m},0]}\|e_1 (t+\theta)\|^2  \nn\\
&&\hskip .0in = - 2(L_f+1)V_1(e_1(t))
\nn\\
&& \hskip .15in 
+2\big[L_f 
\hskip -.03in+  \hskip -.03in
(L_f+1)^2\frac{d}{m}\big] \hskip -.03in
 \max_{\theta\in[-\delta-\frac{2d}{m},0]}V_1(e_1(t+\theta )) 
\eea
Note that  (\ref{m}) implies $2(L_f+  (L_f+1)^2\frac{d }{m})<2(L_f+1)$. 
Thus, it follows from (\ref{dv1'}) and Lemma \ref{lm1} (Halanay inequality)
that the error dynamics (\ref{e1}) is 
GES, i.e., there exist real constants $a,\lambda>0$ such that
\bea\label{inee1}
\|e_1(t)\|\leq a \|\phi_1\|_c e^{-\lambda t}:=\beta_1(\|\phi_1\|_c, t),
\eea
where $\beta_1(s, t)=a s e^{-\lambda t}$, and $\phi_1 \in C([-\delta-\frac{2d}{m},0], \R^n)$ is the initial state of 
the retarded system (\ref{e1}).

For the $i$-th error dynamics
(\ref{ei}), $i=2, \dots, m$, 
consider the Lyapunov function
\bea \label{vi}
V_i(e_i(t)) =  e_i^T(t) e_i(t).
\eea 
Using the error signals 
\bea\label{eit}
e_i(t) = z_i(t) - z_{i-1}  (t + \frac{d}{m} ) \quad \mbox{and}\quad   e_{t}^i:=z_{t}^i-z_{t+\frac{d}{m}}^{i-1},
\eea 
together with (\ref{ei}), (\ref{vi}) 
and  Young inequality,  one has  
\bea \label{dvi}
&&\hskip -.22in  \dot V_i(e_i(t))\leq 2L_f \| e_i(t) \| \| e_{t}^i \|_c - 2(L_f+1)e_i^T(t)e_i (t-\frac{d}{m}) \nn\\
&&\hskip .5in+ 2(L_f+1)e_i^T(t)e_{i-1}(t)\nn\\  
&&\hskip .33in \leq 2L_f \| e_{t}^i \|_c^2  - 2(L_f+1)\|e_i(t)\|^2  \nn\\
&&\hskip .5in +2(L_f+1)\|e_i(t)\| \|e_i(t)-e_i (t-\frac{d}{m})\| \nn\\
&&\hskip .5in
+\epsilon_0\|e_i(t)\|^2+ \frac{(L_f+1)^2}{\epsilon_0}\|e_{i-1}(t)\|^2\nn\\
&&\hskip .33in \leq - 2(L_f+1)\|e_i(t)\|^2 +2L_f \| e_{t}^i \|_c^2  \nn\\
&&\hskip .5in +2(L_f+1)\|e_i(t)\| \int_{t-\frac{d}{m}}^t \|\dot e_i (s)\|\mathrm{d}s \nn\\
&&\hskip .5in
+\epsilon_0\|e_i(t)\|^2+ \frac{(L_f+1)^2}{\epsilon_0}\|e_{i-1}(t)\|^2,
\eea
for $i=2,\cdots,m$,  with
 $\epsilon_0= \frac{1}{2}(1- (L_f+1)^2\frac{d}{m})>0$. 

Akin to the derivation of (\ref{e1cs}), we have the following estimation:
\bea \label{eics}
&&\hskip -.2in \int_{t-\frac{d}{m}}^t \|\dot e_i (s)\|\mathrm{d}s \leq 
 \int_{t-\frac{d}{m}}^t  \hskip -.1in
(L_f+1) \max_{\theta\in[-\delta-\frac{2d}{m},0]}\|e_i (t+\theta)\|\mathrm{d}s \nn\\ 
&& \hskip .7in + \int_{t-\frac{d}{m}}^t (L_f+1) \max_{\theta\in[-\frac{d}{m},0]}\|e_{i-1} (t+\theta)\|\mathrm{d}s \nn\\
&& \hskip -.1in =  \frac{d}{m}(L_f+1)\max_{\theta\in[-\delta-\frac{2d}{m},0]}\|e_i (t+\theta)\| \nn\\
&& \hskip .2in + \frac{d}{m}(L_f+1)\max_{\theta\in[-\frac{d}{m},0]}\|e_{i-1} (t+\theta)\|.
\eea
Using (\ref{dvi}), (\ref{eics}) and Young inequality,  we obtain
\bea \label{dvi'}
&&\hskip -.3in \dot{V}_i(t) \leq - 2(L_f+1)\|e_i(t)\|^2 +2L_f \max_{\theta\in[-\delta-\frac{2d}{m},0]}\|e_i (t+\theta)\|^2  \nn\\
&&\hskip .2in +\frac{2 d}{m} (L_f+1)^2\max_{\theta\in[-\delta-\frac{2d}{m},0]}\|e_i (t+\theta)\|^2 \nn\\
&&\hskip .2in +\epsilon_0\|e_i(t)\|^2
+ \frac{d^2(L_f+1)^4}{m^2\epsilon_0}\max_{\theta\in[-\frac{d}{m},0]}\|e_{i-1} (t+\theta)\|^2 \nn\\
&&\hskip .2in
+\epsilon_0\|e_i(t)\|^2+ \frac{(L_f+1)^2}{\epsilon_0}\|e_{i-1}(t)\|^2  \nn\\
&&\hskip .0in \leq  - 2(L_f+1- \epsilon_0)V_i(e_i(t))
\nn\\
&&\hskip .2in 
+2\big(L_f+  (L_f+1)^2\frac{d}{m}\big) \max_{\theta\in[-\delta-\frac{2d}{m},0]}V_1(e_i(t+\theta )) \nn\\
&&\hskip .2in + c_1 \max_{\theta\in[-\frac{d}{m},0]}\|e_{i-1} (t+\theta)\|^2, 
\eea
where $c_1=\frac{d^2(L_f+1)^4}{m^2\epsilon_0}+\frac{(L_f+1)^2}{\epsilon_0}$.

In view of (\ref{m}) and $\epsilon_0= \frac{1}{2}(1- (L_f+1)^2\frac{d}{m})$, it is clear that $2(L_f+1- \epsilon_0)>2(L_f+  (L_f+1)\frac{d}{m})$. 
This, together with
(\ref{dvi'}) and Lemma \ref{lm3},  leads to the conclusion
 that the error dynamics (\ref{ei}) is  ISS with respect
to $e_{i-1,t}$.
 In other words,  there exist  functions $\beta_i(\cdot,\cdot)$ of class $\mathcal{KL}$ 
  and $\gamma_i(\cdot)$ of class $\mathcal{K}$,
such that
\bea\label{eiss}
\|e_{i}(t)\|\leq \beta_i(\|\phi_i\|_c,t)
\hskip -.03in + \hskip -.03in 
\gamma_i(\sup_{0\leq s\leq t}\|e_{i-1,s}\|_c), \  i=2,\cdots,m
\eea
where $\phi_i$ is the initial state of the retarded system (\ref{ei}).

By (\ref{inee1}), (\ref{eiss}) with $i=2$ and Lemma \ref{lm4}, 
there exists a $\mathcal{KL}$ function $\bar{\beta}_2(\cdot,\cdot)$ such that
\bea\label{inee2}
\|(e_{1}(t), e_2(t))\|\leq \bar{\beta}_2(\|(\phi_1,\phi_2)\|_c,t).
\eea

Recursively, by (\ref{inee2}), (\ref{eiss}) and  Lemma \ref{lm4}, there exists a $\mathcal{KL}$ function $\bar{\beta}_m(\cdot,\cdot)
$ such that
\bea \label{bet1}
\|e(t)\|\leq \bar{\beta}_m(\|\phi\|_c,t), 
\eea
where $e(t)=[e_1^T(t)  \cdots  e_m^T(t)]^T\in \R^{mn}$, and  $\phi=[\phi_1^T  \cdots  \phi_m^T]^T\in C([-\delta-d,0],\R^{mn})$ is the initial state
of the error dynamics (\ref{e1})-(\ref{ei}).

In view of  
 (\ref{er}), we have 
\bea
z_m(t) &=& e_m(t) + z_{m-1} (t + \frac{d}{m} ) \nn\\
z_{m-1}(t) &=& e_{m-1}(t) + z_{m-2} (t + \frac{d}{m} ) \nn\\
&\vdots& \nn\\
z_1(t) &=&  e_1(t) + x (t + \frac{d}{m} ).
\eea
As a consequence,  
\bea \label{zm}
z_m(t) =x(t + d)+ \sum_{j=0}^{m-1} e_{m-j} (t +  \frac{j}{m}d ),
\eea 
Putting (\ref{zm}) and (\ref{u}) together,  we can rewrite the retarded system (\ref{s1}) 
as 
\bea  \label{dx1}
\hskip -.0in \dot{x}(t) =  f (x_t, 
\alpha (x(t)+ v(t)  ) ),
\eea 
where $v(t)=\sum_{j=0}^{m-1} e_{m-j}
 (t -  \frac{m-j}{m}d )$.

From Assumption \ref{ap1}  and 
Lemma \ref{lm2}, it is concluded that the retarded system (\ref{dx1}) is 
ISS with respect to $v(t)$, 
i.e., there exist a function $\beta(\cdot,\cdot)$ of class $\mathcal{KL}$ 
 and a function $\gamma(\cdot)$ of class $\mathcal{K}$, such that
\bea \label{bet2}
\|x(t)\|\leq \beta(\|\psi\|_c,t)+\gamma\big(\sup_{0\leq s\leq t}\|v(s)\|\big), \quad \forall t\geq0,
\eea
where $\psi\in C([-\delta,0], \R^n)$ is the initial state of the retarded system (\ref{s1}).

By (\ref{bet1}), (\ref{bet2}), $v(t)=\sum_{j=0}^{m-1} e_{m-j}
 (t - \frac{m-j}{m}d )$ and Lemma \ref{lm4}, 
 there is a $\mathcal{KL}$ function $\bar{\beta}(\cdot,\cdot)$ such that
\bea\label{bet}
\|(x(t),e(t))\|\leq \bar{\beta}(\|(\psi,\phi)\|_c,t).
\eea 
This, in turn, implies that 
the closed-loop retarded system composed of (\ref{s1}) and  (\ref{u})-(\ref{sq})
 is GAS.
The proof of Theorem \ref{th1} is complete.
\hfill\qed

\begin{remark} \rm
It should be pointed out that the 
 bound $m$ in (\ref{m}) can be further reduced by selecting the gains  of 
sequential predictors (\ref{sq}) appropriately. In fact, 
by setting the gains of sequential predictors smaller, 
we can obtain a smaller  
$m$, as summarized in the conclusion below.
\end{remark}

\begin{theorem}
\label{th2} \rm
Under Assumption  \ref{ap1}, the 
state feedback controller
\bea \label{u'}
u(t) = \alpha (z_m(t))
\eea
with the sequential predictors
\bea \label{sq'}
&&\hskip -.33in \dot{z}_1(t) = f \big(z_{t}^1, u(t-d + \frac{d}{m} ) \big) \hskip -.02in - \hskip -.02in (L_f+\varepsilon)\big(z_1 (t \hskip -.02in - \hskip -.02in \frac{d}{m} ) \hskip -.02in - \hskip -.02in x(t) \big) \nn\\
&&\hskip -.33in \dot{z}_2(t) = f \big(z_{t}^2, u(t-d + \frac{2d}{m} ) \big) \hskip -.02in - \hskip -.02in
(L_f+\varepsilon) \big(z_2  (t \hskip -.02in - \hskip -.02in \frac{d}{m}  ) \hskip -.02in - \hskip -.02in z_1(t) \big) \nn\\
&&\vdots \nn\\
&&\hskip -.33in \dot{z}_m(t) = f (z_{t}^m, u(t) ) \hskip -.02in - \hskip -.02in
(L_f+\varepsilon) \big(z_m  (t \hskip -.02in - \hskip -.02in \frac{d}{m}  ) \hskip -.02in - \hskip -.02in z_{m-1}(t) \big)
\eea
globally asymptotically stabilizes  the retarded 
system (\ref{s1}), where 
$m\in  \mathbb{N}$ satisfies
\bea \label{m'}
m > (L_f+\varepsilon)^2d,
\eea
for any  
 $0<\varepsilon<1$.  
 \hfill \qed
\end{theorem}

The proof  is almost identical to 
that of Theorem \ref{th1} and thus is omitted.

When the retarded system (\ref{s1}) 
contains only  
a point-wise delay in the state, i.e.,
\bea
\dot x(t)=f(x(t),x(t-\delta), u(t-d)), \label{s4}
\eea
the following result is a direct consequence of Theorem \ref{th1} (respectively,  \ref{th2}) and Remark \ref{rk1}.

\begin{corollary} \label{cor1}
If the time-delay nonlinear system (\ref{s4}) satisfies Assumption  \ref{ap1},  the 
state feedback controller
\bea 
u(t) = \alpha (z_m(t))
\eea
with the sequential predictors
\bea \label{sq3}
&&\hskip -.2in \dot{z}_1(t) = f \big(z_1(t),z_1(t-\delta), u(t-d + \frac{d}{m} ) \big) \nn\\
 &&\hskip .3in -(L_f+\varepsilon)\big(z_1 (t \hskip -.02in - \hskip -.02in \frac{d}{m} ) \hskip -.02in - \hskip -.02in x(t) \big) \nn\\
&&\hskip -.2in \dot{z}_2(t) = f \big(z_2(t),z_2(t-\delta), u(t-d + \frac{2d}{m} ) \big) \nn\\
 &&\hskip .3in -
(L_f+\varepsilon) \big(z_2  (t \hskip -.02in - \hskip -.02in \frac{d}{m}  ) \hskip -.02in - \hskip -.02in z_1(t) \big) \nn\\
&&\vdots \nn\\
&&\hskip -.2in \dot{z}_m(t) = f (z_m(t),z_m(t-\delta), u(t) ) \nn\\
 &&\hskip .3in -
(L_f+\varepsilon) \big(z_m  (t \hskip -.02in - \hskip -.02in \frac{d}{m}  ) \hskip -.02in - \hskip -.02in z_{m-1}(t) \big)
\eea
globally asymptotically stabilizes  the time-delay system (\ref{s4}), where 
$m\in \mathbb{N}$ satisfies
\bea \label{m''}
m > (L_f+\varepsilon)^2d, \qquad 
\mbox{for any} \quad 
\varepsilon \in (0,1). 
\eea
\end{corollary}

 \section{Output Feedback via Sequential Predictors}

When the state $x(t)$ is unmeasurable and the delayed output
$y(t)=h(x(t-\tau))$ is the only measurable signal,  it is necessary to investigate
 the stabilization problem by measurement feedback.  In this section, we present an output feedback control scheme
 based on sequential predictors  to compensate large input/output delays in the retarded system (\ref{s1})-(\ref{y}).
 The following result 
  constitutes another contribution of this article

\begin{theorem}
\label{th3} \rm  
If  the retarded system (\ref{s1})-(\ref{y}) with input/output delays satisfies Assumptions \ref{ap1} and \ref{ap2},
there is an output feedback controller
\bea \label{uo}
u(t) = \alpha (\hat{z}_m(t))
\eea
with the sequential predictors
\bea 
&&\hskip -.45in \dot{\hat{z}}_0(t) = F\big(\hat{z}_{t}^0, u(t-d-\tau), y(t)\big) \label{z0}\nn\\
&&\hskip -.45in \dot{\hat{z}}_1(t) = F \big(\hat{z}_{t}^1, u(t-d-\tau + \frac{d+\tau}{m}), h(\hat{z}_1(t))\big) \nn\\
&&\hskip .0in  - (L_{F}+L_{F}L_{h}+1)\big(\hat{z}_1  (t-\frac{d+\tau}{m}  ) - \hat{z}_0(t) \big) \nn\\
&&\hskip -.45in \dot{\hat{z}}_2(t) = F \big(\hat{z}_{t}^2, u(t-d-\tau + \frac{2(d+\tau)}{m} ), h(\hat{z}_2(t)) \big) \nn\\
&&\hskip .0in  - 
(L_{F}+L_{F}L_{h}+1) \big(\hat{z}_2  (t-\frac{d+\tau}{m} ) - \hat{z}_1(t) \big) \nn\\
&&\hskip -.05in \vdots \nn\\
&&\hskip -.45in \dot{\hat{z}}_m(t) = F (\hat{z}_{t}^m, u(t), h(\hat{z}_m(t)) ) \nn\\
&&\hskip .0in  - 
(L_{F}+L_{F}L_{h}+1) \big(\hat{z}_m  (t-\frac{d+\tau}{m}  ) - \hat{z}_{m-1}(t) \big)  \label{sqo}
\eea
such that the time-delay closed-loop system composed of  (\ref{s1})-(\ref{y}) and 
(\ref{uo})-(\ref{sqo})
is  globally asymptotically stable, where 
$\hat{z}_{t}^i\in C([-\delta, 0], \R^n)$, $i=1,\cdots,m$ with  $\hat{z}_{t}^i(\theta):=\hat{z}_{i}(t+\theta)$,
$\forall \theta\in [-\delta, 0]$, and
$m\in \mathbb{N}$ satisfies
\bea \label{mo}
m >(L_{F}+L_{F}L_{h}+1)^2(d+\tau).
\eea
\end{theorem}

{\it Proof:}\   Following 
the proof of Theorem \ref{th1}, we define 
 \bea
z_0(t)=x(t-\tau), \qquad z^0_t=x_{t-\tau}
\eea 
 and the estimate errors 
\bea
&&\hskip -.3in \tilde{e}_0(t)=\hat{z}_0(t)- z_0(t)
\label{e0}\\
&&\hskip -.3in  \tilde{e}_i(t)=\hat{z}_i(t)-\hat{z}_{i-1}(t+\frac{d+\tau}{m}), \quad i=1,\cdots,m. \label{ei'}
\eea
In view of  (\ref{s1}),  (\ref{uo}), and the first equation of (\ref{sqo}), it is not difficult to verify that
\bea
&& \hskip -.53in \dot{\tilde{e}}_0 (t) = F\big(\hat{z}_{t}^0, u(t-d-\tau), y(t)\big) \hskip -.03in - \hskip -.03in f(z_{t}^0, u(t-d-\tau))
\label{de0}
\eea
where 
$y(t)=h(x(t-\tau))=h(z_0(t))$.

From Assumption \ref{ap2}, it is concluded that 
$\hat{z}_0(t)$ generated by the first predictor of (\ref{sqo})
is a global 
exponential estimation of $z_0(t)=x(t-\tau)$. In other words, 
the error dynamics 
(\ref{de0}) is  
GES, i.e., there exist constants $a_0,\lambda_0>0$ such that
\bea\label{inee0}
\|\tilde{e}_0(t) \|\leq a_0 \|\tilde{\phi}_0\|_c e^{-\lambda_0 t},
\eea
where  
$\tilde{\phi}_0 \in C([-\delta-\tau, 0],\ \R^n)$ is the initial state of the  
nonlinear retarded system (\ref{de0}) with input delay.

From (\ref{inee0}), it is clear that 
\bea\label{inee0'}
\hskip -.03in \|
 \tilde{e}_0(t+\frac{d+\tau}{m})\| \hskip -.02in \leq \hskip -.02in  a_0 e^{-\lambda_0 \frac{d+\tau}{m}} \|\tilde{\phi}_0\|_c e^{-\lambda_0 t} \hskip -.03in := \hskip -.03in \tilde{\beta}_0(\|\tilde{\phi}_0\|_c, t),
\eea
where $\tilde{\beta}_0(s, t)=a_0e^{-\lambda_0 \frac{d+\tau}{m}} se^{-\lambda_0 t}$.

Similarly, it is deduced from the second predictor of (\ref{sqo}) and (\ref{ei'}) that 
\bea\label{de1'}
&&\hskip -.25in \dot{\tilde{e}}_1(t)=
 F \big(\hat{z}_{t}^1, \alpha  (\hat{z}_m (t-d-\tau + \frac{d+\tau}{m} ) ), h(\hat{z}_1(t))\big) \nn\\
&& \hskip .0in - F\big(\hat{z}_{t+\frac{d+\tau}{m}}^0, \alpha (\hat{z}_m (t-d-\tau+\frac{d+\tau}{m})), y(t+\frac{d+\tau}{m})\big) \nn\\
&& \hskip .0in  - (L_{F}+L_{F}L_{h}+1) \tilde{e}_1 (t-\frac{d+\tau}{m} ).
\eea
Using the error signals 
 \bea
\tilde{e}_1(t)=\hat{z}_1(t)-\hat{z}_{0}(t+\frac{d+\tau}{m}), \quad \tilde{e}_{t}^1=\hat{z}_{t}^1-\hat{z}_{t+\frac{d+\tau}{m}}^0,
\eea 
and the GLC property of $F,h$,  (\ref{y}), and (\ref{e0})-(\ref{ei'}), we have
\bea \label{f1j}
&&\hskip -.1in \|F \big(\hat{z}_{t}^1, \alpha  (\hat{z}_m (t-d-\tau + \frac{d+\tau}{m} ) ), h(\hat{z}_1(t))\big) \nn\\
&& \hskip .0in - F\big(\hat{z}_{t+\frac{d+\tau}{m}}^0, \alpha (\hat{z}_m (t-d-\tau+\frac{d+\tau}{m})), y(t+\frac{d+\tau}{m})\big)\| \nn\\
&&\hskip -.1in \leq L_{F}\big(\|\tilde{e}_{t}^1\|_c+L_{h}\|\hat{z}_1(t)-x(t-\tau+\frac{d+\tau}{m})\|\big) \nn\\
&&\hskip -.1in \leq  L_{F} \|\tilde{e}_{t}^1\|_c+L_{F} L_{h}\|\hat{z}_1(t)-\hat{z}_0(t+\frac{d+\tau}{m})\| \nn\\
&&\hskip .05in + L_{F} L_{h}\|\hat{z}_0(t+\frac{d+\tau}{m})-x(t-\tau+\frac{d+\tau}{m})\| \nn\\
&&\hskip -.1in =  L_{F} \|\tilde{e}_{t}^1\|_c+L_{F} L_{h}\|\tilde{e}_1(t)\| + L_{F} L_{h}\|
 \tilde{e}_0(t+\frac{d+\tau}{m})\|.
\eea

For the error dynamics (\ref{de1'}), let us consider the quadratic Lyapunov function
\bea
\tilde{V}_1(\tilde{e}_1(t))=\tilde{e}^T_1(t)\tilde{e}_1(t)
\eea
whose derivative along the solution of (\ref{de1'}) satisfies (by (\ref{f1j}) and Young inequality)
\bea \label{DV1}
&&\hskip -.4in \dot{\tilde{V}}_1(\tilde{e}_1(t))\leq 2L_{F}\|\tilde{e}_1(t)\|\|\tilde{e}_{t}^1\|_c+2L_{F}L_{h}\|\tilde{e}_1(t)\|^2 \nn\\
&&\hskip .2in +2L_{F}L_{h}\|\tilde{e}_1(t)\|\|\tilde{e}_0(t +\frac{d+\tau}{m})\| \nn\\
&&\hskip .2in -2(L_{F}+L_{F}L_{h}+1) \tilde{e}^T_1(t)\tilde{e}_1 (t-\frac{d+\tau}{m}) \nn\\
&&\hskip -.2in \leq -2(L_{F}+L_{F}L_{h}+1) \|\tilde{e}_1(t)\|^2 \nn\\
&& \hskip -.1in +2(L_{F}+L_{F}L_{h}+1) \|\tilde{e}_1(t)\|\|\tilde{e}_1(t)-\tilde{e}_1 (t-\frac{d+\tau}{m})\| \nn\\
&&+2(L_{F}+L_{F}L_{h}) \max_{\theta\in[-\delta-2\frac{d+\tau}{m},0]}  \|\tilde{e}_1(t+\theta)\|^2 \nn\\
&&+\epsilon_1\|\tilde{e}_1(t)\|^2+\frac{(L_{F}L_{h})^2}{\epsilon_1} \|
 \tilde{e}_0(t+\frac{d+\tau}{m})\|^2, 
\eea
where  
$\epsilon_1=\frac{1}{2}(1-\frac{(L_{F}+L_{F}L_{h}+1)^2(d+\tau)}{m})>0$.

By  (\ref{de1'})-(\ref{f1j}), the following estimation holds:
\bea\label{ine1}
&&\hskip -.45in \|\tilde{e}_1(t)-\tilde{e}_1 (t-\frac{d+\tau}{m})\|\leq \int_{t-\frac{d+\tau}{m}}^t \|\dot{\tilde{e}}_1(s)\|\mathrm{d}s \nn\\
&& \hskip -.45in \leq (L_{F}+L_{F}L_{h}+1)\frac{d+\tau}{m}\max_{\theta\in[-\delta-2\frac{d+\tau}{m},0]}\|\tilde{e}_1(t+\theta)\| \nn\\
&& \hskip -.25in +L_{F}L_{h} \frac{d+\tau}{m} \max_{\theta\in[-\frac{d+\tau}{m},0]}\|
\tilde{e}_0(t+\frac{d+\tau}{m}+\theta)\|.
\eea
This, together with  Young inequality, results in
\bea\label{E1}
&&\hskip -.1in 2(L_{F}+L_{F}L_{h}+1)  \|\tilde{e}_1(t)
\|\tilde{e}_1(t)-\tilde{e}_1 (t-\frac{d+\tau}{m})\| \nn\\
&&\hskip -.1in \leq 2(L_{F}+L_{F}L_{h}+1)^2\frac{d+\tau}{m} \max_{\theta\in[-\delta-2\frac{d+\tau}{m},0]}\|\tilde{e}_1(t+\theta)\|^2 \nn\\
&&\hskip .1in +2(L_{F}+L_{F}L_{h}+1)L_{F}L_{h}  \nn\\
&&\hskip .5in \cdot  \frac{d+\tau}{m} \|\tilde{e}_1(t)\| \max_{\theta\in[-\frac{d+\tau}{m},0]}\|
 \tilde{e}_0(t+\frac{d+\tau}{m}+\theta)\| \nn\\
&&\hskip -.1in \leq 2(L_{F}+L_{F}L_{h}+1)^2\frac{d+\tau}{m} \max_{\theta\in[-\delta-2\frac{d+\tau}{m},0]}\|\tilde{e}_1(t+\theta)\|^2 \nn\\
&&\hskip .0in + c_2
\max_{\theta\in[-\frac{d+\tau}{m},0]}\|
 \tilde{e}_0(t+\frac{d+\tau}{m}+\theta)\|^2
+ \epsilon_1\|\tilde{e}_1(t)\|^2, 
\eea
where $c_2=\frac{(L_{F}+L_{F}L_{h}+1)^2(L_{F}L_{h})^2(d+\tau)^2}{\epsilon_1m^2}$.

Substituting (\ref{E1}) into  (\ref{DV1}) yields 
\bea \label{DV12}
&&\hskip -.25in \dot{\tilde{V}}_1(\tilde{e}_1(t))\leq
-2(L_{F}+L_{F}L_{h}+1-\epsilon_1) \|\tilde{e}_1(t)\|^2 \nn\\
&&\hskip .45in  +2\Big[L_{F}+L_{F}L_{h}
 \nn\\ && 
+(L_{F}+L_{F}L_{h}+1)^2\frac{d+\tau}{m} \Big] \max_{\theta\in[-\delta-2\frac{d+\tau}{m},0]}  \|\tilde{e}_1(t+\theta)\|^2 \nn\\
&&
+ c_3\max_{\theta\in[-\frac{d+\tau}{m},0]}\|
 \tilde{e}_0(t+\frac{d+\tau}{m}+\theta)\|^2 \nn\\
&&\hskip -.15in = -2(L_{F}+L_{F}L_{h}+1-\epsilon_1) \tilde{V}_1(\tilde{e}_1(t))+2\Big[L_{F}+L_{F}L_{h}  \nn\\
&&\hskip -.1in +(L_{F}+L_{F}L_{h}+1)^2 \frac{d+\tau}{m} \Big]\max_{\theta\in[-\delta-2\frac{d+\tau}{m},0]}  \tilde{V}_1(\tilde{e}_1(t+\theta)) 
\nn\\
&&\hskip -.in
+ c_3\max_{\theta\in[-\frac{d+\tau}{m},0]}\|
 \tilde{e}_0(t+\frac{d+\tau}{m}+\theta)\|^2,
\eea
where $c_3=\frac{ (L_{F}L_{h})^2}{\epsilon_1}+ c_2.$

By (\ref{mo}) and  $\epsilon_1=\frac{1}{2}(1-\frac{(L_{F}+L_{F}L_{h}+1)^2(d+\tau)}{m})$, one has
\bea \label{ine3}
&&\hskip -.2in  2(L_{F}+L_{F}L_{h}+1-\epsilon_1) \nn\\
&&\hskip -.2in  >2\big(L_{F}+L_{F}L_{h}+(L_{F}+L_{F}L_{h}+1)^2\frac{d+\tau}{m} \big). 
\eea
This, together with (\ref{DV12}) and Lemma \ref{lm3}, leads to the conclusion that
 the error dynamics   (\ref{de1'}) 
is ISS, i.e., there exist  functions $\tilde{\beta}_1(\cdot,\cdot)$ of class $\mathcal{KL}$ and $\tilde{\gamma}_1(\cdot)$
of class $\mathcal{K}$, such that
\bea\label{e1iss}
\|\tilde{e}_{1}(t)\|\leq \tilde{\beta}_1(\|\tilde{\phi}_1\|_c,t)+\tilde{\gamma}_1(\sup_{0\leq s\leq t} \|
\tilde{e}_{0,s+\frac{d+\tau}{m}}\|_c),
\eea
where $\tilde{\phi}_1$ is the initial state of the error dynamic system (\ref{de1'}).

By (\ref{e1iss}), (\ref{inee0'}) and Lemma \ref{lm4}, there exists a $\mathcal{KL}$ function $\beta_{1e}(\cdot,\cdot)$
such that  
\bea\label{bet1e}
\|(\tilde{e}_0(t),
\tilde{e}_1(t))\|\leq \beta_{1e}(\|(\tilde{\phi}_0,\tilde{\phi}_1)\|_c,t).
\eea

Following a similar argument as done for the error dynamics (\ref{de1'}), 
we deduce from (\ref{sqo}) and (\ref{ei'}) that 
\bea\label{dei'}
&&\hskip -.26in \dot{\tilde{e}}_i(t)=F \big(\hat{z}_{t}^i, \alpha (\hat{z}_m (t-d-\tau + \frac{i(d+\tau)}{m} ) ), h(\hat{z}_i(t)) \big) \nn\\
&& \hskip .2in  -F \big(\hat{z}_{t+\frac{d+\tau}{m}}^{i-1}, \alpha (\hat{z}_m (t-d-\tau + \frac{i(d+\tau)}{m} ) ), \nn\\
&& \hskip 2.0in  h(\hat{z}_{i-1}(t+\frac{d+\tau}{m})) \big) \nn\\
&& \hskip .2in  - 
(L_{F}+L_{F}L_{h}+1) \tilde{e}_i (t-\frac{d+\tau}{m} )  \nn\\
&& \hskip .2in  + (L_{F}+L_{F}L_{h}+1) \tilde{e}_{i-1} (t), \qquad\  i=2,\cdots,m.
\eea
By the use of
\bea
 \tilde{e}_i(t)=\hat{z}_i(t)-\hat{z}_{i-1}(t+\frac{d+\tau}{m}), \quad \ \tilde{e}_{t}^i=\hat{z}_{t}^i-\hat{z}_{t+\frac{d+\tau}{m}}^{i-1},
\eea 
the GLC property of $F,h$  and  (\ref{ei'}), one has
\bea\label{fij}
&&\hskip -.1in  \|F \big(\hat{z}_{t}^i, \alpha  (\hat{z}_m \big(t-d-\tau + \frac{i(d+\tau)}{m}  ) ), h(\hat{z}_i(t)) \big) \nn\\
&& \hskip .0in  -F \big(\hat{z}_{t+\frac{d+\tau}{m}}^{i-1}, \alpha  (\hat{z}_m  (t-d-\tau + \frac{i(d+\tau)}{m} )  ), 
\nn\\
&& \hskip 1.9in 
 h(\hat{z}_{i-1}(t+\frac{d+\tau}{m}))  \big)\| \nn\\
&&\hskip -.1in \leq L_{F}\big(\|\tilde{e}_{t}^i\|_c+L_{h}\|\hat{z}_i(t)-\hat{z}_{i-1}(t+\frac{d+\tau}{m})\|\big) \nn\\
&&\hskip -.1in = L_{F} \|\tilde{e}_{t}^i\|_c+L_{F} L_{h}\|\tilde{e}_i(t)\|, \qquad\quad  i=2,\cdots,m.
\eea

For the error dynamics (\ref{dei'}), 
consider the quadratic Lyapunov functions 
\bea
\tilde{V}_i(\tilde{e}_i)=\tilde{e}_i^T\tilde{e}_i, \qquad i=2,\cdots,m.
\eea
A direct calculation gives (by (\ref{fij}) and Young inequality)
\bea\label{DVi1}
&&\hskip -.4in \dot{\tilde{V}}_i(\tilde{e}_i(t))\leq -2(L_{F}+L_{F}L_{h}+1) \tilde{e}^T_i(t)\tilde{e}_i (t-\frac{d+\tau}{m})\nn\\
&&\hskip .2in +2L_{F}\|\tilde{e}_i(t)\|\|\tilde{e}_{t}^i\|_c+2L_{F}L_{h}\|\tilde{e}_i(t)\|^2 \nn\\
&& \hskip .2in +2(L_{F}+L_{F}L_{h}+1) \|\tilde{e}_i(t)\| \|\tilde{e}_{i-1}(t)\| \nn\\
&&\hskip -.25in \leq -2(L_{F}+L_{F}L_{h}+1) \|\tilde{e}_i(t)\|^2 \nn\\
&&+2(L_{F}+L_{F}L_{h}+1) \|\tilde{e}_i(t)\|\|\tilde{e}_i(t)-\tilde{e}_i (t-\frac{d+\tau}{m})\| \nn\\
&&+2(L_{F}+L_{F}L_{h}) \max_{\theta\in[-\delta-2\frac{d+\tau}{m},0]}  \|\tilde{e}_i(t+\theta)\|^2 \nn\\
&&+\epsilon_1\|\tilde{e}_i(t)\|^2+\frac{(L_{F}+L_{F}L_{h}+1)^2}{\epsilon_1} \|\tilde{e}_{i-1}(t)\|^2.
\eea

Similar to the estimation of (\ref{ine1}), we deduce from (\ref{dei'})-(\ref{fij}) that 
\bea\label{ine2}
&& \hskip -.45in \|\tilde{e}_i(t)-\tilde{e}_i (t-\frac{d+\tau}{m})\| \leq \int_{t-\frac{d+\tau}{m}}^t \|\dot{\tilde{e}}_i(s)\|\mathrm{d}s \nn\\
&& \hskip -.45in \leq (L_{F}+L_{F}L_{h}+1)\frac{d+\tau}{m}\max_{\theta\in[-\delta-2\frac{d+\tau}{m},0]}\|\tilde{e}_i(t+\theta)\| \nn\\
&& \hskip -.3in +(L_{F}+L_{F}L_{h}+1) \frac{d+\tau}{m} \max_{\theta\in[-\frac{d+\tau}{m},0]}\|\tilde{e}_{i-1}(t+\theta)\|.
\eea
Akin to the derivation of  inequality 
(\ref{DV12}), it follows from 
 (\ref{DVi1})-(\ref{ine2}) and Young inequality that 
\bea \label{DVi2}
&&\hskip -.5in \dot{\tilde{V}}_i(\tilde{e}_i(t))\leq -2(L_{F}+L_{F}L_{h}+1-\epsilon_1) \|\tilde{e}_i(t)\|^2 \nn\\
&&\hskip .1in +2\big(L_{F}+L_{F}L_{h}+(L_{F}+L_{F}L_{h}+1)^2\frac{d+\tau}{m} \big) \nn\\ 
&& \hskip .1in \cdot \max_{\theta\in[-\delta-2\frac{d+\tau}{m},0]}  \|\tilde{e}_i(t+\theta)\|^2 \nn\\
&& \hskip .1in
+ c_4\max_{\theta\in[-\frac{d+\tau}{m},0]}\| \tilde{e}_{i-1}(t+\theta)\|^2 \nn\\
&&\hskip -.4in = -2(L_{F}+L_{F}L_{h}+1-\epsilon_1) \tilde{V}_i(\tilde{e}_i(t)) \nn\\
&&\hskip -.2in +2\big(L_{F}+L_{F}L_{h}+(L_{F}+L_{F}L_{h}+1)^2\frac{d+\tau}{m} \big) \nn\\
&&\hskip -.2in  \cdot \max_{\theta\in[-\delta-2\frac{d+\tau}{m},0]}  \tilde{V}_i(\tilde{e}_i(t+\theta)) 
\nn\\
&&\hskip -.2in
+ c_4\max_{\theta\in[-\frac{d+\tau}{m},0]}\| \tilde{e}_{i-1}(t+\theta)\|^2,
\eea
where $c_4=\frac{(L_{F}+L_{F}L_{h}+1)^2}{\epsilon_1}+ 
\frac{(L_{F}+L_{F}L_{h}+1)^4(d+\tau)^2}{\epsilon_1m^2}$.

According to (\ref{DVi2}) and Lemma \ref{lm3}, we conclude that the error dynamics (\ref{dei'}) 
is ISS, i.e., there exist functions $\tilde{\beta}_i(\cdot,\cdot)$ of class $\mathcal{KL}$  
and  
$\tilde{\gamma}_i(\cdot)$ of class $\mathcal{K}$,
 such that 
\bea\label{eiss'}
\|\tilde{e}_{i}(t)\|\leq \tilde{\beta}_i(\|\tilde{\phi}_i\|_c,t)
\hskip -.03in + \hskip -.03in 
\tilde{\gamma}_i(\sup_{0\leq s\leq t}\|\tilde{e}_{i-1,s}\|_c), i=2,\cdots,m
\eea
where $\tilde{\phi}_i$ is the initial state of the error dynamic system (\ref{dei'}).

Based on  (\ref{bet1e}) and (\ref{eiss'}), 
one can apply
Lemma \ref{lm4} repeatedly and conclude that the  error dynamic systems (\ref{de0}), (\ref{de1'}) and (\ref{dei'}) with $i=2,\cdots,m$ 
are  GAS. That is,  there is a $\mathcal{KL}$ function $\beta_{e}(\cdot,\cdot)$ such that
\bea \label{bet3}
\|\tilde{e}(t)\|\leq \beta_{e}(\|\tilde{\phi}\|_c,t), \quad \forall t\geq0.
\eea
where $\tilde{e}(t)=[\tilde{e}_0^T(t) \ \tilde{e}_1^T(t) \ \cdots \ \tilde{e}_m^T(t)]^T\in \R^{(m+1)n}$ and  $\tilde{\phi}\in C([-\delta-\tau,0], \R^{(m+1)n})$ is the initial state
of the error dynamic systems (\ref{de0}), (\ref{de1'}) and (\ref{dei'}). 

In view of  
$\tilde{e}_i(t)$ defined by (\ref{e0})-(\ref{ei'}), we obtain
\bea
&&\hskip -.4in \hat{z}_m(t) = \tilde{e}_m(t) + \hat{z}_{m-1} (t + \frac{d+\tau}{m} ) \nn\\
&&\hskip -.4in \hat{z}_{m-1}(t) = \tilde{e}_{m-1}(t) + \hat{z}_{m-2} (t + \frac{d+\tau}{m} ) \nn\\
&&\hskip -.2in \vdots \nn\\
&&\hskip -.4in \hat{z}_1(t) =  \tilde{e}_1(t) + \hat{z}_{0}(t + \frac{d+\tau}{m} ) \nn\\
&&\hskip -.4in \hat{z}_0(t) =  \tilde{e}_0(t) + x (t -\tau),
\eea
which, in turn, result in  
\bea \label{zm2}
\hat{z}_m(t) =x(t + d)+ \sum_{j=0}^{m} \tilde{e}_{m-j} (t +  \frac{j}{m}(d+\tau) ).
\eea 
Substituting (\ref{zm2}) into 
(\ref{uo}), we see that  the 
retarded system (\ref{s1}) under the output feedback controller (\ref{uo}) can be expressed as
\bea  \label{dx}
\hskip -.0in \dot{x}(t) =  f (x_t, 
\alpha (x(t)+ \tilde{v}(t)  ) ):= g(x_t, \tilde{v}(t)),
\eea 
where $\tilde{v}(t)=\sum_{j=0}^{m} \tilde{e}_{m-j}
 (t -d +  \frac{j}{m}(d+\tau))$,

By Assumption \ref{ap1}  and Lemma \ref{lm2}, 
the retarded system (\ref{dx}) is 
ISS.  Equivalently,  
 there exist 
 functions $\hat{\beta}(\cdot,\cdot)$ of class $\mathcal{KL}$  and  
 $\hat{\gamma}(\cdot)$ of class $\mathcal{K}$, such that
\bea \label{bet4}
\|x(t)\|\leq \hat{\beta}(\|\psi\|_c,t)+\hat{\gamma}\big(\sup_{0\leq s\leq t}\|\tilde{v}(s)\|\big), \quad \forall t\geq0,
\eea
where $\psi\in C([-\delta,0], \R^n)$ is the initial state of the retarded system (\ref{s1}) with input delay.

Finally,  it is concluded from (\ref{bet3}), (\ref{bet4}), $\tilde{v}(t)=\sum_{j=0}^{m} \tilde{e}_{m-j}
 (t -d +  \frac{j}{m}(d+\tau))$ and Lemma \ref{lm4} that 
the time-delay closed-loop system composed by (\ref{s1})-(\ref{y}) and (\ref{uo})-(\ref{sqo}) 
  is  GAS.   This completes the proof of  Theorem \ref{th3}.
\hfill\qed

Similar to the case of state feedback,  
 the lower bound of $m$ in (\ref{mo}) can also be reduced 
by choosing the gains of sequential predictors (\ref{sqo})
 suitably.
 The following result whose proof is a strong reminiscent of  
  that of Theorem \ref{th3} and thus is  omitted.

\begin{theorem}
\label{th4} \rm 
Under Assumptions  \ref{ap1}-\ref{ap2}, the 
output feedback controller
\bea \label{uo'}
u(t) = \alpha (\hat{z}_m(t))
\eea
with the sequential predictors
\bea 
&&\hskip -.45in \dot{\hat{z}}_0(t) = F\big(\hat{z}_{t}^0, u(t-d-\tau), y(t)\big), 
\nn\\
&&\hskip -.45in \dot{\hat{z}}_1(t) = F \big(\hat{z}_{t}^1, u(t-d-\tau + \frac{d+\tau}{m} ), h(\hat{z}_1(t))\big) \nn\\
&&\hskip .0in  - (L_{F}+L_{F}L_{h}+\varepsilon)\big(\hat{z}_1  (t-\frac{d+\tau}{m}  ) - \hat{z}_0(t) \big) \nn\\
&&\hskip -.45in \dot{\hat{z}}_2(t) = F \big(\hat{z}_{t}^2, u(t-d-\tau + \frac{2(d+\tau)}{m} ), h(\hat{z}_2(t)) \big) \nn\\
&&\hskip .0in  - 
(L_{F}+L_{F}L_{h}+\varepsilon) \big(\hat{z}_2  (t-\frac{d+\tau}{m} ) - \hat{z}_1(t) \big) \nn\\
&&\hskip -.05in \vdots \nn\\
&&\hskip -.45in \dot{\hat{z}}_m(t) = F (\hat{z}_{t}^m, u(t), h(\hat{z}_m(t)) ) \nn\\
&&\hskip .0in  - 
(L_{F}+L_{F}L_{h}+\varepsilon) \big(\hat{z}_m  (t-\frac{d+\tau}{m}  ) - \hat{z}_{m-1}(t) \big)  \label{sqo'}
\eea
globally asymptotically stabilizes  the retarded  
system (\ref{s1})-(\ref{y}) with input/output delays, as long as  
$m\in \mathbb{N}$ satisfies
\bea \label{mo'}
m >(L_{F}+L_{F}L_{h}+\varepsilon)^2(d+\tau), 
\eea
for any 
$\varepsilon\in (0,1)$. 
\hfill \qed
\end{theorem}

Consider a special 
 case when the retarded system (\ref{s1}) 
 has only point-wise delay in the state, i.e., 
\bea
&&\hskip -.2in \dot x(t)=f(x(t),x(t-\delta), u(t-d)) \label{s5} \\
&&\hskip -.2in 
y(t)=h(x(t-\tau)), \label{y'}
\eea
In this case, Assumption \ref{ap2} reduces to 
the  condition below.

\begin{assumption} \label{ap2'} \rm
For the input/output retarded 
system (\ref{s5})-(\ref{y'}) without output delay,
 i.e., $\tau=0$,  there  exists a GLC function
$F: \R^n \times \R^n \times \R^p \times \R^q\to \R^n$ with a Lipschitz constant $L_{F}>0$ and $F(0,0,0,0)=0$, such that
\bea 
\dot{\hat{x}}(t) = F(\hat{x}(t), \hat{x}(t-\delta), u(t-d), y(t)) 
\eea
is a global  
exponential 
observer 
of the time-delay
 system (\ref{s5}), i.e., 
$\|x(t)-\hat x(t)\|\leq \lambda_1\|x_0-\hat{x}_0\|_ce^{-\lambda_2t}$, $\forall x_0,\hat{x}_0\in C([-\delta,0], \R^n)$, where 
$\lambda_1$ and $\lambda_2$ are positive 
constants. 
\end{assumption}

Then, the following result can be deduced from 
Theorems \ref{th3} and \ref{th4}.

\begin{corollary} \label{cor2}
If the time-delay system (\ref{s5})-(\ref{y'}) with input/output delays satisfies Assumptions \ref{ap1} and \ref{ap2'},
there is an output feedback controller
\bea \label{uo''}
u(t) = \alpha (\hat{z}_m(t))
\eea
with the sequential predictors
\bea 
&&\hskip -.3in \dot{\hat{z}}_0(t) = F\big(\hat{z}_{0}(t),\hat{z}_{0}(t-\delta), u(t-d-\tau), y(t)\big) \label{z0}\nn\\
&&\hskip -.3in \dot{\hat{z}}_1(t) = F \big(\hat{z}_{1}(t),\hat{z}_{1}(t-\delta), u(t-d-\tau + \frac{d+\tau}{m}), h(\hat{z}_1(t))\big) \nn\\
&&\hskip .2in  - (L_{F}+L_{F}L_{h}+\varepsilon)\big(\hat{z}_1  (t-\frac{d+\tau}{m}  ) - \hat{z}_0(t) \big) \nn\\
&&\hskip -.3in \dot{\hat{z}}_2(t) = F \big(\hat{z}_{2}(t),\hat{z}_{2}(t-\delta), u(t-d-\tau + \frac{2(d+\tau)}{m} ), h(\hat{z}_2(t)) \big) \nn\\
&&\hskip .2in  - 
(L_{F}+L_{F}L_{h}+\varepsilon) \big(\hat{z}_2  (t-\frac{d+\tau}{m} ) - \hat{z}_1(t) \big) \nn\\
&&\hskip -.05in \vdots \nn\\
&&\hskip -.3in \dot{\hat{z}}_m(t) = F (\hat{z}_{m}(t),\hat{z}_{m}(t-\delta), u(t), h(\hat{z}_m(t)) ) \nn\\
&&\hskip .2in  - 
(L_{F}+L_{F}L_{h}+\varepsilon) \big(\hat{z}_m  (t-\frac{d+\tau}{m}  ) - \hat{z}_{m-1}(t) \big)  \label{sqo''}
\eea
such that the time-delay closed-loop system composed of  (\ref{s5})-(\ref{y'}) and 
(\ref{uo''})-(\ref{sqo''})
is  
GAS, where 
$m\in \mathbb{N}$ satisfies
\bea \label{mo''}
m >(L_{F}+L_{F}L_{h}+\varepsilon)^2(d+\tau),
\eea
for any  
 $0<\varepsilon\leq 1$. 
\end{corollary}

\begin{remark}\label{rk2}
As a consequence of Theorems \ref{th1} and \ref{th3},   
the following results follows immediately for a class of 
retarded systems in a strict feedback form, i.e.,  
\bea
&&\dot{x}_i(t)=x_{i+1}(t)+\phi_{i}(x^{1}_t,\cdots,x^{i}_t), \quad i=1,\cdots, n-1 \nn\\
&&\dot{x}_n(t)=u(t-d)+\phi_n
(x^{1}_t,\cdots,x^{n}_t) \label{lts} \\
&&\hskip .1in y(t)=x_1(t-\tau), \label{lto}
\eea
where the mappings $\phi_i: C([-\delta,0], \R^i)\to \R$, $i=1,\cdots,n$ are GLC with $\phi_{i}(0,\cdots,0)=0$.

When $d=0$, using the backstepping design method, one can design a linear feedback $u=Kx$ which renders the 
retarded system  (\ref{lts}) with $d=0$ GES. In other words,   Assumption \ref{ap1} is satisfied. 
 By Theorem \ref{th1}, the  state feedback controller $u(t)=Kz_m(t)$ 
 with the sequential predictors (\ref{sq}) makes
 the retraded system (\ref{lts}) with  $d>0$ GAS.

In the case of output feedback,  one can design a set of gains $l_1,l_2, \cdots, l_n$ such that 
\bea
&&\hskip -.5in \dot{\hat{x}}_i(t)=\hat{x}_{i+1}(t)+\phi_{i}(\hat{x}^{1}_t,\cdots,\hat{x}^{i}_t)-l_i(y(t)-\hat{x}_1(t)), \nn\\
&& \hskip 1.8in i=1,\cdots, n-1 \nn\\
&&\hskip -.5in \dot{\hat{x}}_n(t)=u(t-d)+\phi_n
(\hat{x}^{1}_t,\cdots,\hat{x}^{n}_t)-l_n(y(t)-\hat{x}_1(t))
\eea
is a GES observer for the input/output retarded system (\ref{lts})-(\ref{lto}) with $\tau=0$, i.e.,
Assumption \ref{ap2} holds. According to Theorem \ref{th3}, the sequential predictor based output feedback controller 
$u(t)=K\hat{z}_m(t)$ with (\ref{sqo}) 
GAS the retarded system (\ref{lts})-(\ref{lto}) with input/output delays $d,\tau>0$.
\end{remark}

\section{Extensions to  ISS retarded systems with input/output delays 
}

We now discuss how the previous results can be extended to a  broader  
class of retarded systems 
under the following conditions which own a motivation from the 
papers \cite{mm17,bm21}, the survey \cite{Ant23a},
 and the works \cite{Ant23b,Ant24,mw17,mw18} 
on the ISS of time-delay nonlinear systems 
as well as infinite-dimensional systems.

\begin{assumption}\label{ap3}
For the vector field $f$ in  (\ref{s1}), there exists a constant $L_{f}>0$ such that
\bea \label{Lff}
\| f(\phi, u) - f(\varphi, u)\|\leq L_f  \| \phi- \varphi \|_c, 
\eea
 $\forall \phi,  \varphi\in C([-\delta,0], \R^n)$ and $\forall u\in \R^p$. 
\end{assumption}

\begin{assumption} \label{ap4} \rm
For the retarded system (\ref{s1}) without input delay, i.e., $d=0$,
 there exists a state feedback law $u = \alpha(x)$ with $\alpha(0)=0$, which is 
 locally Lipschitz continuous (LLC),  such that
\bea
\dot{x}(t) = f(x_t, \alpha(x(t)+\mu(t))) \label{s3}
\eea
is ISS with respect to the perturbation $\mu(t)$.
\end{assumption}

Under the condition (\ref{Lff}),  it is not difficult to verify  that 
the retarded nonlinear control system (\ref{s1}) is forward complete, 
permitting one to compensate the large input delay by predictor based state and/or output feedback.
 With the aid of Assumptions \ref{ap3} and \ref{ap4}, the following result that refines Theorem \ref{th1} or \ref{th2} can be 
obtained.

\begin{theorem}
\label{th5} \rm  
Suppose that the retarded system (\ref{s1}) with  input delay 
satisfies Assumptions \ref{ap3} and \ref{ap4}. Then,  the 
state feedback controller 
\bea 
u(t) = \alpha (z_m(t))
\eea
with the sequential predictors (\ref{sq})-(\ref{m}) or  (\ref{sq'})-(\ref{m'})
renders  
the 
retarded 
system (\ref{s1}) with input delay GAS. 
\end{theorem}

{\it Proof:} \ 
The proof  is a strong reminiscent of the argument used for the proof 
of Theorem \ref{th1}.  In what follows, we  shall highlight only the major difference.
 
First of all, by Assumption \ref{ap3},
the estimations of $f\big(z_{t}^1, \alpha (z_m (t-d + \frac{d}{m} ) )\big)- f\big(z_{t+\frac{d}{m}}^0, \alpha (z_m (t-d+ \frac{d}{m} ) )\big)$ in (\ref{e1}) and $f\big(z_{t}^i, \alpha (z_m (t-d + \frac{id}{m} ) )\big)- f\big(z_{t+\frac{d}{m}}^{i-1}, \alpha (z_m (t-d+ \frac{id}{m} ) )\big)$ in (\ref{ei}) are given by  
\bea\label{diff1}
&&\hskip -.3in \big\|f\big(z_{t}^1, \alpha (z_m (t-d + \frac{d}{m} ) )\big) \nn\\
&&\hskip -.15in - f\big(z_{t+\frac{d}{m}}^0, \alpha (z_m (t-d+ \frac{d}{m} ) )\big) \big\| \leq L_f\|e_{t}^1\|_c
\eea
and
\bea\label{diff2}
&&\hskip -.3in \big\|f\big(z_{t}^i, \alpha (z_m (t-d + \frac{id}{m} ) )\big) \nn\\
&&\hskip -.15in - f\big(z_{t+\frac{d}{m}}^{i-1}, \alpha (z_m (t-d+ \frac{id}{m} ) )\big) \big\| \leq L_f\|e_{t}^i\|_c.
\eea
Based on (\ref{diff1}) and (\ref{diff2}),  the arguments from (\ref{v1}) to (\ref{dx1}) remain unchanged.

Secondly, by Assumption \ref{ap4},  the retarded system (\ref{dx1}) is 
ISS with respect to $v(t)$, and hence the  inequality (\ref{bet2}) also fulfilled.  

Finally, 
in view of  (\ref{bet1}), (\ref{bet2}) and Lemma \ref{lm4},  there is a  function $\bar{\beta}(\cdot,\cdot)$ of class
$\mathcal{KL}$ such that
(\ref{bet}) holds. This, in turn, leads to Theorem \ref{th5}.
 \hfill\qed

 \begin{remark} It is worth pointing out that even when the retarded system (\ref{s1}) 
degenerates to the case without state delay, 
the proof of (\ref{bet}) is substantially different from
those given in \cite{mm17} and \cite{bm21}, which rely on the construction of
  Lyapunov-Krasovskii (LK) functionals and the LK  functional method.  
For the nonlinear retarded system (\ref{s1}), it is, however, very difficult 
to construct a Lyapunov-Krasovskii functional in an explicit form. 
In addition, the construction of the $\mathcal{KL}$ function $\bar{\beta}(\cdot,\cdot)$ in (\ref{bet}) is also 
different from the ones presented in \cite{mm17} and \cite{bm21}. 
Finally and most notably,  it  should 
 be  emphasized 
 that the ISS argument adopted in eq. (29) of \cite{mm17}, in particular, 
 the term $\gamma_b(2m|e|_{[t-h,t-h/m]})$ in eq. (29) 
 should be $\gamma_b(2m|e|_{[t_0-h,t-h/m]})$, 
 which makes  
 the construction of $\beta_d (q,r)$ in eq. (32)  of \cite{mm17}
potentially invalid. To circumvent the obstacle just described, 
 we construct  in this paper 
the $\mathcal{KL}$ function $\bar{\beta}(\cdot,\cdot)$ in (\ref{bet}) 
based on the proof of 
Lemma \ref{lm4} given in Appendix.
 \end{remark}

 \begin{remark}
Compared with  the existing Lyapunov-Krasovskii functional method,
 the advantage of the Lyapunov-Halanay method presented in this paper is that:
It needs not to construct  Lyapunov-Krasovskii functionals which may be
 a daunting task, if not completely impossible, 
for the nonlinear retarded systems with input/output delays.
Instead, we develop a  Lyapunov-Halanay method in this paper based on
the quadratic Lyapunov functions for the error dynamics, to conduct  
the stability analysis  of the closed-loop system that is of the retarded form mixed with 
 point-wise delays induced by  delayed input/output. 
 \end{remark}

Similar to the case of state feedback,   
Assumption \ref{ap2} can be relaxed 
 by the following  slightly weaker condition 
when studying the problem of compensation of large input/output delays for the nonlinear retarded system  (\ref{s1})-(\ref{y}) 
by predictor based output feedback.

\begin{assumption} \label{ap5} \rm
For the input/output  
system (\ref{s1})-(\ref{y}) without output delay,
 i.e., $\tau=0$, 
 assume that  there  exists a 
 continuous functional
$F: C([-\delta,0], \R^n) \times \R^p \times \R^q\to \R^n$ with $F(0,0,0)=0$ and 
\bea \label{LF'}
\| F(\phi, u, y) - F(\varphi, u, \bar{y})\|\leq L_F (\| \phi- \varphi \|_c+\|y-\bar{y}\|), 
\eea
$\forall y,\bar{y}\in \R^q$, $\forall u\in \R^p$, and $\forall \phi,  \varphi\in C([-\delta,0], \R^n)$  such that
\bea 
\dot{\hat{x}}(t) = F(\hat{x}_t, u(t-d), y(t)) 
\eea
is a global  asymptotic 
observer 
of the retarded 
 system (\ref{s1}), i.e., 
$\|x(t)-\hat x(t)\|\leq \beta(\|x_0-\hat{x}_0\|_c,t)$, $\forall x_0,\hat{x}_0\in C([-\delta,0], \R^n)$, where 
$\beta(\cdot,\cdot)$ is a function of class $\mathcal{KL}$ and  
$L_{F}>0$ is a Lipschitz constant.
\end{assumption}

\begin{theorem}
\label{th6} \rm 
If the retarded system (\ref{s1})-(\ref{y}) with input/output delays satisfies Assumptions \ref{ap3},
\ref{ap4} and \ref{ap5}, 
the dynamic output feedback controller
\bea  
u(t) = \alpha (\hat{z}_m(t))
\eea
with the sequential predictors  
(\ref{sqo})-(\ref{mo}) or 
(\ref{sqo'})-(\ref{mo'})
globally asymptotically stabilizes  the nonlinear  
retarded system (\ref{s1})-(\ref{y}).
\end{theorem}

{\it Proof:} \  The proof  is akin to that of Theorem \ref{th3} and hence  only the difference is highlighted below.
 
Firstly,   the retarded system (\ref{s1}) is forward complete under Assumption \ref{ap3}.
 It is concluded from Assumption \ref{ap5} that 
$\hat{z}_0(t)$ generated by the first predictor of (\ref{sqo})
is a global asymptotic 
 estimation of $z_0(t)=x(t-\tau)$, i.e., (\ref{inee0}) is replaced by
  \bea\label{inee0''}
\|\tilde{e}_0(t) \|\leq \beta_0(\|\tilde{\phi}_0\|_c, t),
\eea
for a $\mathcal{KL}$ function $\beta_0(\cdot,\cdot)$. 
This, in turn, implies the existence of  a $\mathcal{KL}$ function $\tilde{\beta}_0(\cdot,\cdot)$ such that
(\ref{inee0'}) holds. 

Using the condition (\ref{LF'}), one can show that the estimations in (\ref{f1j}) and (\ref{fij}) are still true. 
As a result,
the arguments from (\ref{de1'}) to (\ref{dx}) remain effective.

Secondly, it is deduced from Assumption \ref{ap4} that  the retarded system (\ref{dx}) is 
ISS with respect to $\tilde{v}(t)$. Thus,   the  inequality (\ref{bet4}) also holds. 

Lastly,   Theorem \ref{th6}  is established based on  
(\ref{bet3}), (\ref{bet4}) and Lemma \ref{lm4}.
 \hfill\qed

The following example illustrates the application of Theorem \ref{th5}  to 
a class of 
nonlinear systems with large input 
 delay which is ISS but not globally exponentially stabilizable (GES) by  any 
 smooth state feedback .

\begin{example}  
 Consider the  scalar non-affine 
   system with  input delay 
  \bea\label{cex}
  \dot x(t)=\frac{\sin x^5(t)}{1+x^4(t)}+u^3(t-d).  
  \eea
It is easy to verify that $\frac{\sin x^5}{1+x^4}$ is GLC in $x$, and hence 
 the nonlinear system (\ref{cex}) satisfies
 Assumptions \ref{ap3}.    As a result, the input-delayed system (\ref{cex}) is forward complete.

In the absence of input delay, i.e., $d=0$,   the nonlinear system (\ref{cex}) reduces to 
 \bea\label{exu3}
 \dot x=\frac{\sin x^5}{1+x^4}+u^3.
 \eea
 We claim that  there is {\it no any smooth} controller $u=\alpha(x)$ with $\alpha(0)=0$, such that
 the closed-loop system
   \bea \label{ex0}
  \dot x=\frac{\sin x^5}{1+x^4}+\alpha^3(x)
 \eea
 is GES. 
 
 The claim is proved by  
 contradiction.
 If there is 
 a smooth controller $u=\alpha(x)$ with $\alpha(0)=0$,
 making the system (\ref{ex0}) GES.  Then, the closed-loop system (\ref{ex0}) is also LES. 
 Equivalently, 
the  linearized system of (\ref{ex0}) is LES. This 
contradicts to the fact that the linearization of (\ref{ex0}) is
$\dot x=0$  which is clearly not LES.

Next, it is claimed that the state feedback controller $u=\alpha (x)=-x$  ensures the ISS property of the 
nonlinear system  (\ref{exu3})  
with respect to an additive input signal $\mu$ when $u=\alpha (x+\mu)$.  

To this end, 
 we observe that 
(\ref{exu3}) with $u=-(x+\mu)$ is described by 
   \bea\label{ex1} 
  \dot x= 
  \frac{\sin x^5}{1+x^4}-(x+\mu)^3
  \eea 
Choosing the Lyapunov function $V(x)=\frac{1}{2}x^2$, we see that  along the solution of system (\ref{ex1}),
 \bea\label{dVx}
 \dot V(x)= \frac{x\sin x^5}{1+x^4}-x^4-3x^3\mu-3x^2\mu^2-x\mu^3.   
 \eea
Let $\rho(x):=1-\frac{\sin x^5}{x^3(1+x^4)}$. 
Because
 $\frac{\sin x^5}{x^3(1+x^4)}<1$ and 
$\lim_{x\to 0}\rho(x)=1=\lim_{x\to 
\infty}\rho(x)$, 
  the constant $c:=\inf_{x\in\R} \rho(x)>0$ and 
  $$1-\frac{\sin x^5}{x^3(1+x^4)}\ge c, \qquad \forall x\in \R.$$
This, in turn,  leads to 
\bea\label{dVx1}
\frac{x\sin x^5}{1+x^4}-x^4\le -cx^4, \quad \forall x\in \R. 
\eea
Applying (\ref{dVx1}) and Young inequality,  we deduce from (\ref{dVx}) that 
 \bea
&& \dot V(x)\leq -cx^4+3|x|^3|\mu|+3|x|^2|\mu|^2+|x||\mu|^3 \nn\\
 && \hskip .35in \leq -\frac{c}{2}x^4+b|\mu|^4, \quad  \mbox{for some}\ \ b>0.
 \eea
Thus,  the nonlinear system  (\ref{ex1}) is ISS with respect to $\mu$.  

In conclusion,  the  
 nonlinear system (\ref{cex}) with input delay 
 satisfies Assumptions \ref{ap3}-\ref{ap4} 
 but violates Assumption \ref{ap1}.
 Although Theorem \ref{th1} cannot be applied to solve the stabilization problem of 
(\ref{cex})  by predictor based state feedback, 
Theorem \ref{th5} is still applicable and can be used to compensate the large input delay of the
nonlinear system (\ref{cex}), by means of the sequential predictors developed in this paper.
\end{example}

\section{
Control of  the Pendulum System with State and Input/Output Delays 
by Sequential Predictors}

For the purpose of illustration, 
we 
apply in this section the main results of this paper to solve the feedback stabilization problem 
of the pendulum system in the presence of delays in state, input and output.
  
Consider the pendulum system \cite{mm17}. 
When the damping torque of the pendulum system depends on the angular velocity with a delay $\delta$
(for example, an angular velocity sensor has a delay $\delta$, resulting in the use of the angular velocity signal 
from $\delta$ time earlier for damping torque), 
and the control signal has a delay $d$,  the 
pendulum system is modeled by
\bea \label{pe}
J\ddot{\theta}(t) + \zeta\dot{\theta}(t-\delta) - Mgl\sin\theta(t) = u(t-d),
\eea
where $J=Ml^2$ is the moment of inertia,  $\zeta$ is the damping coefficient, $g$ 
is the gravity constant,
$l$ is the pendulum length, 
and $M$ is the pendulum mass. In this case study,  we assume that
the parameters $M=0.1 kg, l=10 m, g=9.8 m/s^2$, and  $\zeta=0.5$. 

Introduce the states $x_1(t)=\theta(t)$, $x_2(t)=\dot\theta(t)$, and the delayed measurement signal $y(t)=x_1(t-\tau)$.
Then, the  state-space model of the pendulum system (\ref{pe}) is represented by
 \bea
&&\hskip -.45in \dot{x}_1(t)=x_2(t) \nn\\
&&\hskip -.45in \dot{x}_2(t)=\hskip -.03in \frac{1}{Ml^2}u(t-d) - \frac{\zeta}{Ml^2} x_2(t-\delta)+\frac{g}{l}\sin x_1(t) \nn \\
&&\hskip -.38in y(t)=x_1(t-\tau). \label{pen} 
\eea
Obviously,  the pendulum system (\ref{pen}) is GLC with a Lipschitz constant $L_f=1$.

 The control objective is to design state and output feedback controllers that ensure global asymptotic stability 
 of the pendulum system (\ref{pen}) at the 
 equilibrium $(0,0)$,  in the presence of  delays in the state, input and output.

\subsection{State Feedback Case}

When the input delay $d=0$ in (\ref{pen}), 
one can design a delay-free, linear state 
feedback controller, for instance, 
$$ u(t)=\alpha(x(t))=
 -25(x_1(t)+x_2(t)),$$
 which renders the pendulum system (\ref{pen}) without input delay 
 GES.  Thus, Assumption \ref{ap1} is fulfilleded.

By 
Corollary \ref{cor1}, one can design a 
state feedback controller 
\bea \label{penu}
&&\hskip -.5in u(t)=\alpha(z_m(t)) =
-25(z_{m,1}(t)+z_{m,2}(t))
\eea
with the sequential predictors 
 \bea\label{sqf}
&&\hskip -.36in 
\begin{bmatrix}
   \dot{z}_{i,1}(t) \\
   \dot{z}_{i,2}(t)
\end{bmatrix}
\hskip -.03in  =\hskip -.03in  \begin{bmatrix}
 z_{i,2}(t) \\
 \frac{g}{l}\sin z_{i,1}(t)-\frac{\zeta}{Ml^2} z_{i,2}(t-\delta)+\frac{1}{Ml^2}u(t-d + \frac{id}{m})
 \end{bmatrix} \nn\\
&&\hskip .35in
-1.3\begin{bmatrix}
z_{i,1} (t-\frac{d}{m} ) - z_{i-1,1}(t) \\
z_{i,2} (t-\frac{d}{m} ) - z_{i-1,2}(t)
 \end{bmatrix} \nn\\
 &&\hskip .0in i=1,2,\cdots,m 
 \eea
where $[z_{i,1}(t) \ z_{i,2}(t)]^T:=z_i(t)\in\R^2$, $i=1,\cdots,m$ and
$[z_{0,1}(t) \ z_{0,2}(t)]^T=[x_1(t) \ x_2(t)]^T=x(t)\in\R^2$.

The simulations shown in Fig. 1 and Fig. 2 were carried out with the input delay $d=2$, state delay $\delta=1$, 
$m=4>1.3^2\times 2$, and the initial condition  $(x_1(s), x_2(s))=(1, 0)$, 
$(z_{i,1}(s), z_{i,2}(s))=(0.2, 0.1)$, $i=1,\cdots,4$,  $\forall s\in[-2,0]$.
The simulation  results of Figs. 1 and 2  have validated  
 the effectiveness of  
 the sequential predictors based controller (\ref{penu})-(\ref{sqf}).
\begin{figure}[htb]
 \vspace{-.1in}
\centering
\hskip -.1in
\rotatebox{360}{\scalebox{0.63}[0.8]{\includegraphics{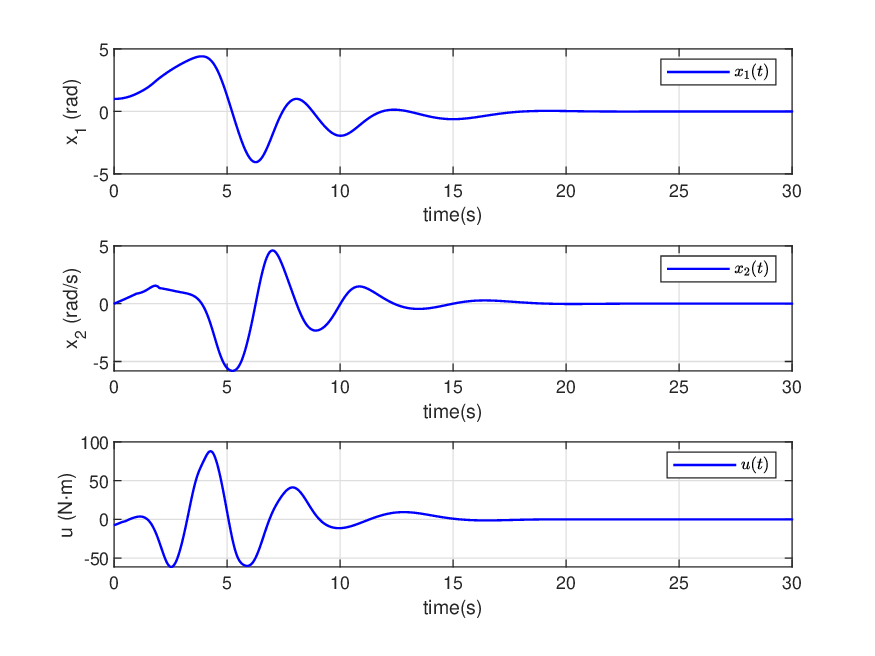}
}}
\vspace{-.1in}
	\center{\small Fig. 1: \  \
		Transient
response of the time-delay closed-loop 
system (\ref{pen})-(\ref{sqf}) 
with $d=2,\delta=1$, and $m=4$. 
}\end{figure}

\begin{figure}[htb]
 \vspace{-.1in}
\centering
\hskip -.1in
\rotatebox{360}{\scalebox{0.63}[0.75]{\includegraphics{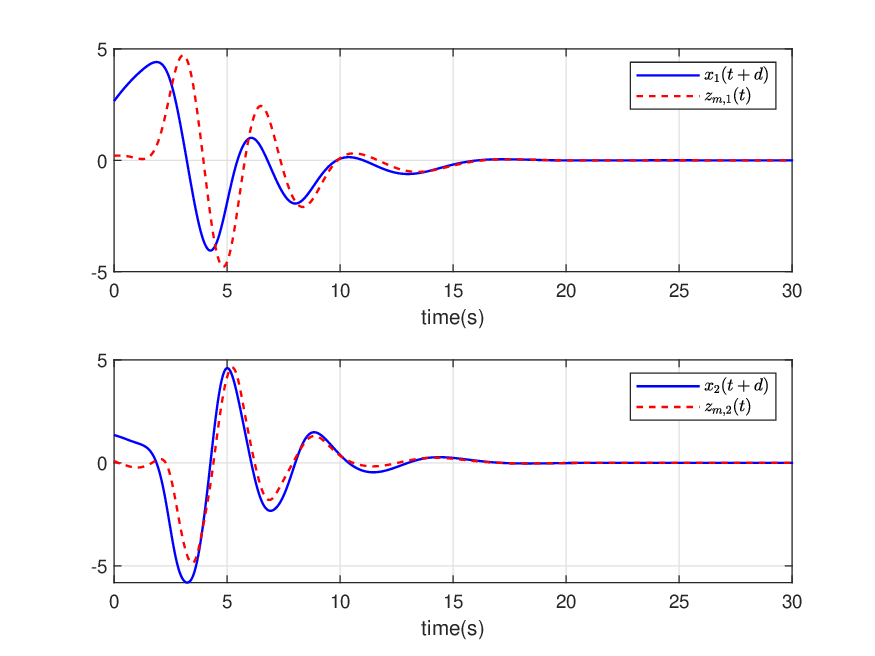}
}}
\vspace{-.1in}
	\center{\small Fig. 2: \  \
	The actual state $x(t+d)$ and its approximate predictor 
$z_m(t)$, where $d=2,m=4$. 
}\end{figure} 

In the case when the delays are increased to $d=4 $ and $\delta=2$, respectively, 
Figs. 3 and 4 give the simulation results that further demonstrate
  the effectiveness of the sequential predictors based controller (\ref{penu})-(\ref{sqf}) 
  with $m=8$ under the same initial conditions.
\begin{figure}[htb]
 \vspace{-.1in}
\centering
\hskip -.1in
\rotatebox{360}{\scalebox{0.63}[0.8]{\includegraphics{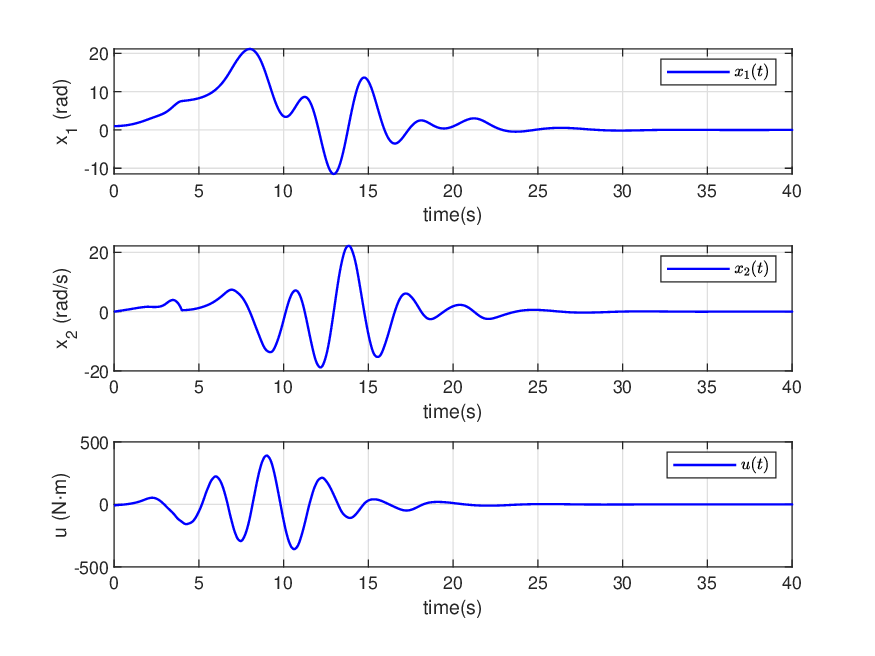}
}}
\vspace{-.1in}
	\center{\small Fig. 3: \  \
		Transient
response of the time-delay closed-loop  system (\ref{pen})-(\ref{sqf}) 
with $d=4,\delta=2$, and $m=8$. 
}\end{figure}

\begin{figure}[htb]
 \vspace{-.1in}
\centering
\hskip -.1in
\rotatebox{360}{\scalebox{0.63}[0.75]{\includegraphics{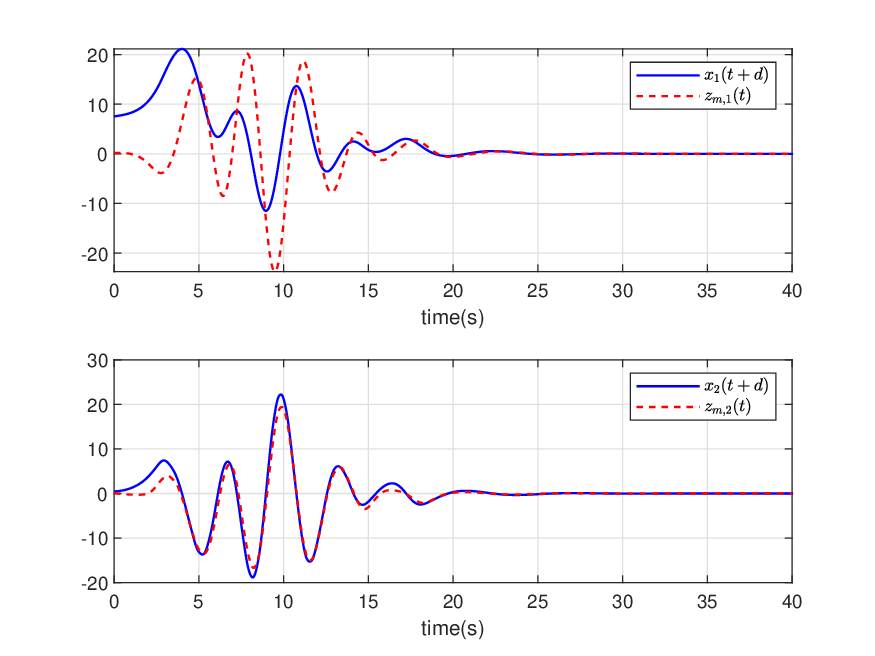}
}}
\vspace{-.1in}
	\center{\small Fig. 4: \  \
	The actual state $x(t+d)$ and its approximate predictor 
$z_m(t)$, where $d=4,m=8$. 
}\end{figure}

\subsection{Output Feedback Case}

When there is no output delay (i.e., $\tau=0$) in the 
pendulum system (\ref{pen}),  we first verify that
Assumption \ref{ap2} holds.  In fact,  
the input/output system (\ref{pen}) with $\tau=0$ 
 can be rewritten as 
 \bea
&&\hskip -.0in \dot{x}(t)=Ax(t)+Bu(t-d)+g_0(x(t),x(t-\delta)) \nn\\
&&\hskip .0in y(t)= 
Cx(t), \label{py} 
\eea
where $x(t)=[x_1(t) \ x_2(t)]^T\in \R^2$, $A=\left[\begin{smallmatrix}0 & 1 \\ 0 & 0\end{smallmatrix}\right]$, 
$B=[0 \ 1/(Ml^2)]^T$, $C=[1 \ 0]$, and the vector field 
\bea
g_0(x(t),x(t-\delta))=\Big[ 0 \ \ \frac{g}{l}\sin x_1(t)-\frac{\zeta}{Ml^2} x_2(t-\delta) \Big]^T. \  
\label{g0}
\eea

For the time-delay system (\ref{py}), 
consider an observer of the form 
\bea
&&\hskip -.35in \dot{\hat{x}}(t)=A\hat{x}(t)+Bu(t-d)+g_0(\hat{x}(t),\hat{x}(t-\delta))  \nn\\
&&\hskip .1in -L(y(t)-\hat{x}_1(t)) \nn\\
&&\hskip -.1in := F(\hat{x}(t),\hat{x}(t-\delta),u(t-d),y(t)), \label{pob}
\eea
where $L=[l_1 \ l_2]^T
\in \R^2$ is the observer gain.

Let $e(t)=[e_1(t) \ e_2(t)]^T=x(t)-\hat x(t)$ be  the estimate error. 
For the error dynamics 
\bea
 \dot{e}(t) = (A-LC)e(t)+ G_0(e(t), e(t-\delta))
\label{err}
\eea
with $G_0(\cdot)=g_0(x(t), x(t-\delta))-g_0(\hat{x}(t),\hat{x}(t-\delta))$  satisfying  
 $||G_0(e(t), e(t-\delta)||\le \frac{g}{l} |e_1(t)| + \frac{\zeta}{Ml^2}|e_2(t-\delta)|$,
one can design 
an observer gain $L$ such that the error dynamic system (\ref{err}) is exponentially stable. In
the simulation study, the observer gain is chosen as $L=[1 \ 1.5]^T$.

Finally, it is easy to check 
that the vector field $F$ in (\ref{pob}) is GLC with a Lipschitz constant $L_{F}=1.5$.
In view of  $y=h(x(t-\tau))=x_1(t-\tau)$ in (\ref{pen}), the mapping $h$ is GLC with a Lipschitz constant $L_h=1$.
 By 
Corollary \ref{cor2}, the output feedback
controller 
\bea \label{penu'}
&&\hskip -.5in u(t)=\alpha(\hat{z}_m(t))= -25(\hat{z}_{m,1}(t)+\hat{z}_{m,2}(t)),
\eea 
with the sequential predictors 
(\ref{z0}), 
 renders the time-delay pendulum system (\ref{pen}) with  input/output delays 
 GAS  if the parameter $m$ is chosen properly.

In the simulations,  we pick the parameter $\varepsilon=0.1$ in (\ref{sqo''}). 
As a result, the sequential predictors for the pendulum system (\ref{pen}) with delays are given by 
\bea 
&&\hskip -.35in \dot{\hat{z}}_0(t) = 
A\hat{z}_{0}(t)+Bu (t-d-\tau)+g_0(\hat{z}_{0}(t),\hat{z}_{0}(t-\delta)) \nn\\
&&\hskip .15in -L(y(t)-\hat{z}_{01}(t)) \nn 
\\
&&\hskip -.35in \dot{\hat{z}}_1(t) = 
A\hat{z}_{1}(t)+Bu(t-d-\tau + \frac{d+\tau}{m}) \nn\\
&&\hskip .15in +g_0(\hat{z}_{1}(t),\hat{z}_{1}(t-\delta)) - 3.1\big(\hat{z}_1  (t-\frac{d+\tau}{m}  ) - \hat{z}_0(t) \big) \nn\\
&&\hskip -.35in \dot{\hat{z}}_2(t) = 
A\hat{z}_{2}(t)+Bu(t-d-\tau + \frac{2(d+\tau)}{m}) \nn\\
&&\hskip .15in +g_0(\hat{z}_{2}(t),\hat{z}_{2}(t-\delta)) - 
3.1 \big(\hat{z}_2  (t-\frac{d+\tau}{m} ) - \hat{z}_1(t) \big) \nn\\
&&\hskip .05in \vdots \nn\\
&&\hskip -.35in \dot{\hat{z}}_m(t) = 
A\hat{z}_{m}(t)+Bu(t) +g_0(\hat{z}_{m}(t),\hat{z}_{m}(t-\delta)) \nn\\
&&\hskip .15in - 
3.1 \big(\hat{z}_m  (t-\frac{d+\tau}{m}  ) - \hat{z}_{m-1}(t) \big),  \label{psq}
\eea
where  $\hat{z}_i(t)=[\hat{z}_{i,1}(t) \ \hat{z}_{i,2}(t)]^T\in \R^2$, $i=0,1,\cdots,m$, and
$g_0(\cdot)$ is defined by (\ref{g0}).

The simulations given in Fig. 5 and Fig. 6  are conducted under the following parameters:
the input/output delays $d=\tau=1$ and  the state delay $\delta=2$,
 $m=20>3.1^2\times 2$ (based on (\ref{mo'})), and  the initial condition  
 $(x_1(s), x_2(s))=(1, 0)$, $(\hat{z}_{i,1}(s), \hat{z}_{i,2}(s))=(0.2, 0.1)$, $i=0,1,\cdots, 20$, $\forall s\in[-2,0]$.
The simulation  results shown in Figs. 5-6 have illustrated  
that the output feedback controller  (\ref{penu'}) with the sequential predictors (\ref{psq}) is capable of compensating
the input/output delays in the pendulum system (\ref{pen}) effectively, and, in turn, 
making the  time-delay closed-loop 
system GAS.

\begin{figure}[htb]
 \vspace{-.1in}
\centering
\hskip -.1in
\rotatebox{360}{\scalebox{0.63}[0.8]{\includegraphics{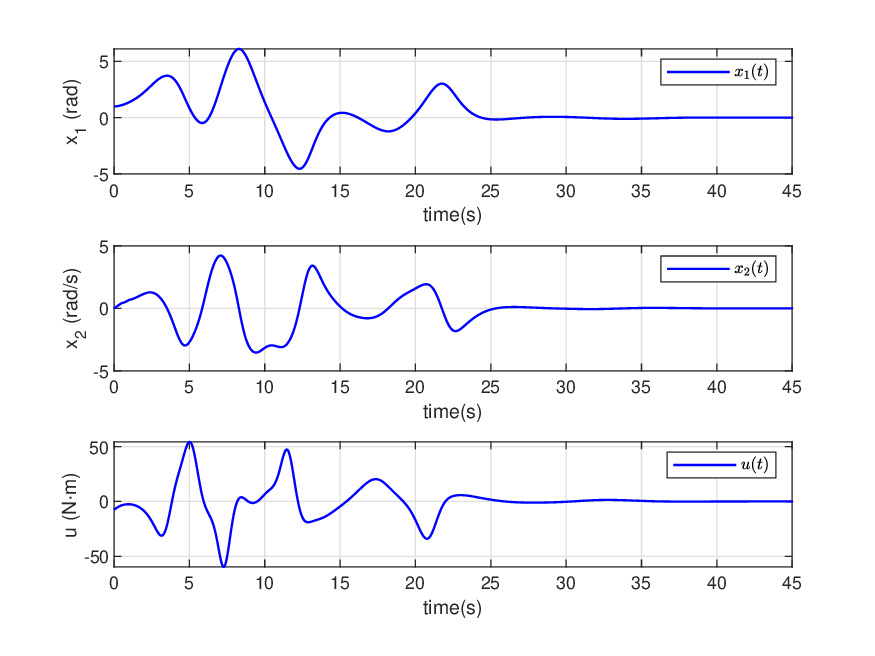}
}}
\vspace{-.1in}
	\center{\small Fig. 5: \  \
		Transient
response of the  time-delay closed-loop system (\ref{pen}), (\ref{penu'})-(\ref{psq})
with $d=\tau=1,\delta=2$, and $m=20$. 
}\end{figure} 

\begin{figure}[htb]
 \vspace{-.1in}
\centering
\hskip -.1in
\rotatebox{360}{\scalebox{0.63}[0.75]{\includegraphics{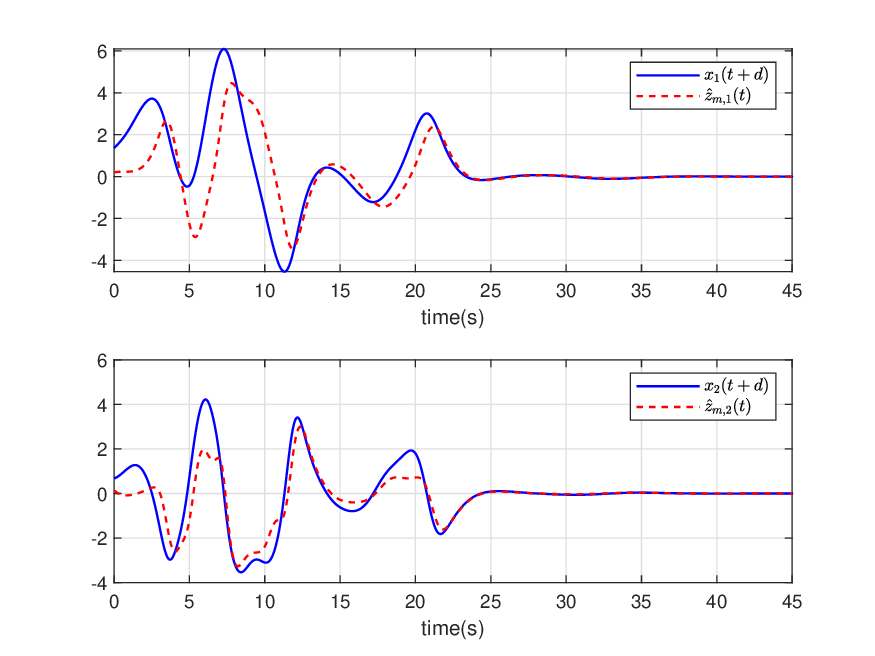}
}}
\vspace{-.1in}
	\center{\small Fig. 6: \  \
	The actual state $x(t+d)$ and its approximate predictor 
 $\hat{z}_m(t)$, where $d=1,m=20$. 
}\end{figure}

\section{Conclusion}

We have  addressed the problems of 
global asymptotic stabilization (GAS) 
by predictor based state and output feedback, respectively, 
for a class of nonlinear retarded systems with long input/output delays.
When the system state is available for control design, 
 a state feedback 
 control strategy based on sequential predictors has been developed, leading to a solution to the 
 GES problem under the  
 global Lipschitz continuity (GLC)   
and global exponential stabilizability (GES) of the nonlinear retarded system {\it without input delay.}
When only the  
delayed output is measurable for the design of controller,   
an output feedback 
 control strategy based on sequential predictors has also been presented, resulting in a solution to the 
 GES problem under the condition that  the nonlinear retarded system {\it without output delay}
  permits  a global exponential observer and global exponential stabilizability via state feedback.
 In addition, both state and output feedback results have also been extended to a broader class of nonlinear retarded systems
 that {\it may not be GES by state feedback} but satisfy  {\it global asymptotic} stabilizability/observability
 and  appropriate  ISS conditions. 
 
 In contrast to the existing LK functional approach in the literature, 
 a distinguished feature of  the proposed control methods in this work lies in the fact that the  
 stability analysis relies on the Lyapunov-Halanay like inequality and  needs not to 
 construct  Lyapunov-Krasovskii functionals 
 -- a daunting task 
 when dealing with nonlinear retarded systems with large input/output delays.
The Lyapunov-Halanay like approach adopted in this research 
employs simple
 quadratic Lyapunov functions for the error dynamics to conduct  
the stability analysis of the closed-loop retarded system mixed with 
 point-wise delays induced by input/output. 
  The proposed state and output feedback control schemes based on sequential predictors
   are capable of compensating large input/output delays and
     ensuring  
     GAS of the retarded closed-loop systems, as demonstrated by 
 the  applications to  a  
   pendulum system with delays in the state,  input and output.

\section*{Appendix: Proof of Lemma \ref{lm4}}

Denote $r = \|(\psi, \phi)\|_c$.  Clearly,  $r \geq \|\psi\|_c$ and $r \geq \|\phi\|_c$. 

By (\ref{b1}) and Proposition 2 of \cite{Ant23a}, there is a $\mathcal{KL}$ function $\bar{\beta}_1(\cdot, \cdot)$ and a $\mathcal{K}$ function $\bar{\gamma}_1(\cdot)$  such that
\bea \label{A'}
\|x_t\|_c \leq \bar{\beta}_1(\|\psi\|_c, t) + \bar{\gamma}_1 ( \sup_{0 \leq s \leq t} \|e_{s}\|_c ).
\eea
This, in turn,  implies that for any $s \in [0,t]$,
\bea \label{A}
\|x_t\|_c \leq \bar{\beta}_1(\|x_{s}\|_c, t-s) + \bar{\gamma}_1 ( \sup_{s \leq \tau \leq t} \|e_{\tau}\|_c ).  
\eea

On the other hand, there exists a $\mathcal{KL}$ function $\bar{\beta}_2(\cdot, \cdot)$ such that (due to (\ref{b2}))
\bea
\|e_t\|_c \leq \bar{\beta}_2(\|\phi\|_c, t). 
\label{B}
\eea

Let us  fix $t > 0$. Applying (\ref{A}) with $s = t/2$ results in
\bea
\|x_t\|_c \leq \bar{\beta}_1\Big( \|x_{t/2}\|_c, \frac{t}{2} \Big) + \bar{\gamma}_1 \Big(\sup_{t/2 \leq \tau \leq t} \|e_\tau\|_c\Big). \label{2}
\eea
From (\ref{b1}), it follows that
\bea
\|x_{t/2}\|_c \leq \beta_1\Big( \|\psi\|_c, \frac{t}{2} \Big) + \gamma_1\Big( \sup_{0 \leq \tau \leq t/2} \|e_\tau\|_c \Big). 
\label{3}
\eea
Using (\ref{B}) and the decreasing property of the $\mathcal{KL}$ function with respect to its second variable,
we have
\bea
&&\hskip -.3in \sup_{0 \leq \tau \leq t/2} \|e_\tau\|_c \leq \sup_{0 \leq \tau \leq t/2} \bar{\beta}_2(\|\phi\|_c, \tau) \leq \bar{\beta}_2(\|\phi\|_c, 0) \label{4} \nn\\
&&\hskip -.3in
\sup_{t/2 \leq \tau \leq t} \|e_\tau\|_c \leq \sup_{t/2 \leq \tau \leq t} \bar{\beta}_2(\|\phi\|_c, \tau) \leq \bar{\beta}_2 (\|\phi\|_c, \frac{t}{2}). 
\ \
\label{5}
\eea

Substituting (\ref{3})--(\ref{5}) into (\ref{2}), and using $\|\psi\|_c \leq r\), \(\|\phi\|_c \leq r$ with monotonicity of $\beta_1(\cdot,\cdot)$,  $\gamma_1(\cdot)$, $\bar{\beta}_1(\cdot,\cdot)$, 
and $\bar{\gamma}_1(\cdot)$, we obtain
\bea
\|x_t\|_c \leq \bar{\beta}_1 \Big(\beta_1 ( r, \frac{t}{2}  ) + \gamma_1 ( \bar{\beta}_2(r, 0) ), \frac{t}{2} \Big) + 
 \bar{\gamma}_1 \big( \bar{\beta}_2 ( r, \frac{t}{2} )\big). \label{6}
\eea
Furthermore, it is deduce from (\ref{B}) that
\bea
\|e_t\|_c \leq \bar{\beta}_2(r, t). \label{7}
\eea
Hence, 
\bea\label{99}
\|(x_t, e_t)\|_c \leq \|x_t\|_c + \|e_t\|_c \leq \beta(r, t),
\eea
where
\bea
&&\hskip -.3in \beta(r, t) = \bar{\beta}_1 \Big( \beta_1 ( r, \frac{t}{2}  ) + \gamma_1 ( \bar{\beta}_2(r, 0)  ), \frac{t}{2}  \Big) \nn\\
&&\hskip .23in +  \bar{\gamma}_1 \big( \bar{\beta}_2 ( r, \frac{t}{2} )\big)  + \bar{\beta}_2(r, t).
\eea
By construction, 
$\beta(\cdot,\cdot)$ is a $\mathcal{KL}$ function. 

As a consequence,  (\ref{99}) leads to
\bea
\|(x(t), e(t))\| \leq \beta ( \|(\psi, \phi)\|_c, t ).
\eea
This completes the proof of Lemma \ref{lm4}. \hfill\qed

\begin{IEEEbiography}[{\includegraphics[width=1in,height=1.25in,clip,keepaspectratio]{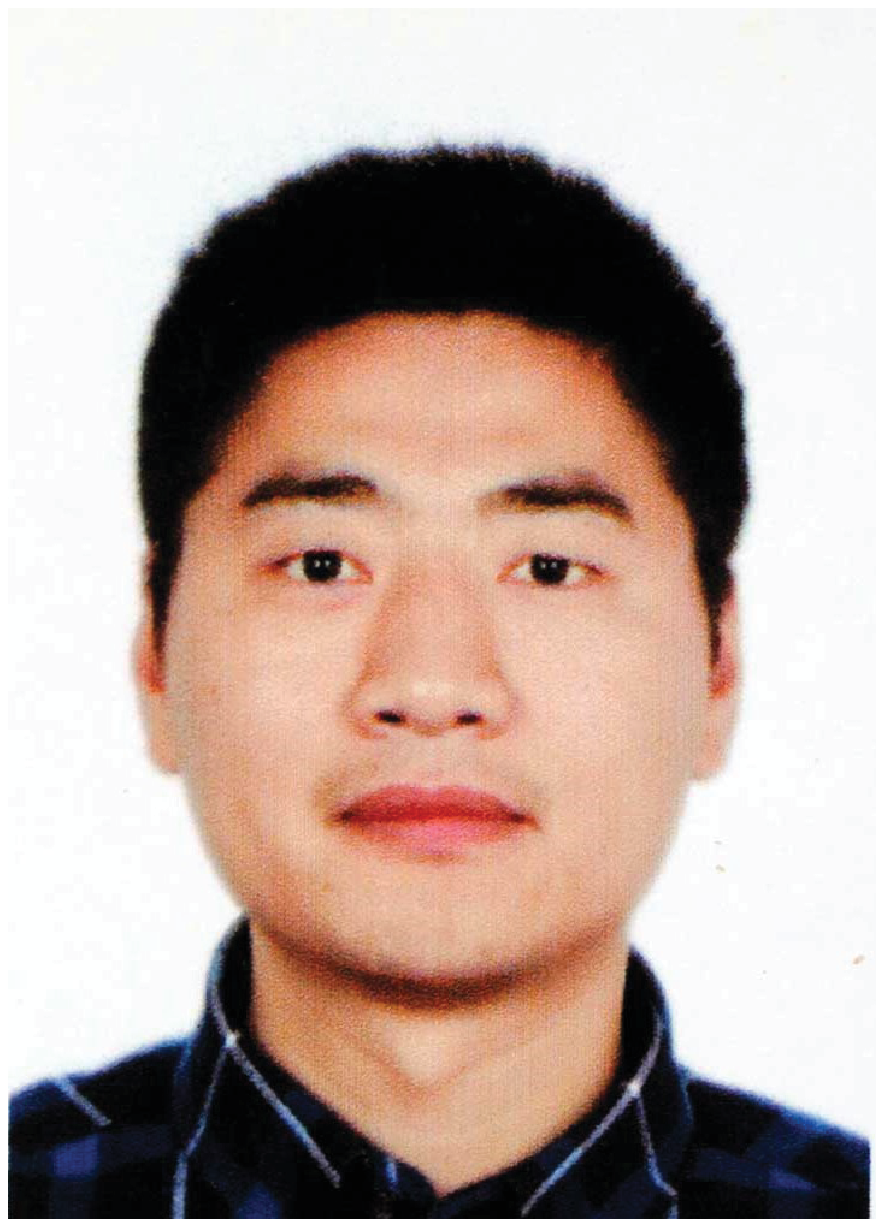}}]
{Xin Yu}
 received the B.S. degree in mathematics and 
 M.S. degree in control theory
from Qufu Normal University, 
China, in 2006 and 2009, respectively,
and PhD degree in control theory and control engineering from
Southeast University, 
China, in 2012.

From 2012 to 2022, he was an Associate Professor in the School of Electrical and Information
Engineering, Jiangsu University, Zhenjiang, China. Since 2023, he has been with the School
of Electrical Engineering and Automation, Jiangsu Normal University, where he is currently a Professor.
His research interests include nonlinear control, homogeneous systems theory, time-delay systems, stochastic control systems, and sampled-data control. 
\end{IEEEbiography}


\begin{IEEEbiography}[{\includegraphics[width=1in,height=1.25in,clip,keepaspectratio]{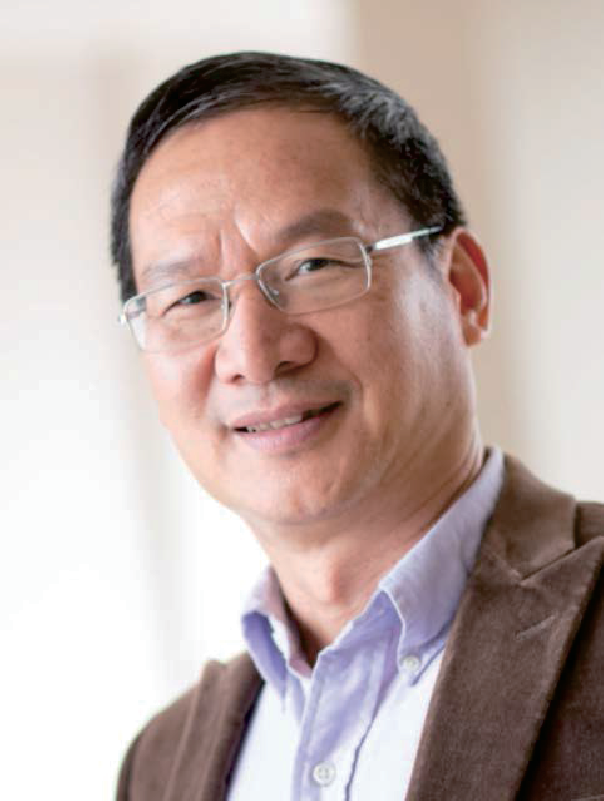}}]
{Wei Lin}
received the D.Sc. and M.S. degrees in systems science and mathematics from Washington
University, St. Louis, MO, USA, in 1993 and 1991,  M.S. degree from Huazhong University
of Science and Technology in 1986, and B.S.  from Dalian University of
Technology in 1983,  both in electrical engineering.  During 1986-1989, he was a Lecturer in the Department of Mathematics, Fudan University, Shanghai, China.  From 1994 to 1995, he was a Post-doctor and   
 a Visiting Assistant Professor in Washington University.
Since Spring of 1996, he has been a Professor
in the Department of Electrical, Computer, and Systems  Engineering,
Case Western
Reserve University, Cleveland, Ohio, USA.

Dr. Lin's research interests include nonlinear control,
time-delay systems, stochastic stability and stochastic control,
 homogeneous systems theory, estimation and adaptive control, under-actuated mechanical
and nonholonomic 
systems, renewable energy, power systems and smart grids.
 In these areas,  he has published   
 a number of peer reviewed
 papers  in  journals and conferences.
  More details can be found at   https://engineering.case.edu/research/labs/nonlinear-control-systems/about

Dr. Lin was a recipient of  the NSF CAREER Award, 
 the Warren E. Rupp Endowed
Professorship,  the Robert Herbold Faculty Fellow Award, the JSPS Fellow and IEEE Fellow.
He  served
as  an Associate Editor of  IEEE Trans. on  on Automa. Contr. (1999-2002),
a Guest Editor of  the special issue on ``New Directions in Nonlinear Control''
in IEEE Trans. on  on Automa. Contr. (2003), an Associate Editor of Automatica (2003-2005),
 a Subject Editor (2005-2010)  
of  Int. J. of Robust and Nonlinear Control, and an Associate Editor of  Journal of Control Theory and Applications (2005-2008),  on the Board of Governors of IEEE Control Systems Society (2003-2005),
 and Vice Program Chair of 2001 CDC (Short Papers) and 2002 CDC (Invited Sessions).
\end{IEEEbiography}

\end{document}